\documentclass[10pt,final,twocolumn]{IEEEtran}
\usepackage[nocompress]{cite}
\usepackage{graphicx,subfigure}
\usepackage{amsthm}
\usepackage{algorithmic}
\usepackage{array}
\usepackage{stfloats}
\usepackage{amsfonts}
\usepackage{tikz}
\newtheorem{Lemma}{Lemma}
\newtheorem{proposition}{Proposition}
\usepackage{empheq}
\usepackage{booktabs,array}
\usepackage{multirow}
\usepackage{tikz}

\renewcommand{\eqref}[1]{(\ref{#1})}
\usepackage{epstopdf}
\usepackage{steinmetz}
\DeclareMathOperator{\E}{\mathbb{E}}

\begin{document}

\title{ Analysis of Statistical QoS in Half Duplex and Full Duplex Dense Heterogeneous Cellular Networks}

\author{Alireza Sadeghi,~\IEEEmembership{Student Member, IEEE,}
        Michele Luvisotto,~\IEEEmembership{Student Member, IEEE,}
        Farshad Lahouti,~\IEEEmembership{Senior Member, IEEE,}
        Michele Zorzi~\IEEEmembership{Fellow, IEEE,}%
\IEEEcompsocitemizethanks{
	\IEEEcompsocthanksitem  A. Sadeghi is with the Electrical and Computer 
	Engineering Department, University of Minnesota, 200 Union Street SE, 4-174 
	Keller Hall, Minneapolis, MN 55455-0170. E-mail: sadeg012@umn.edu
	\IEEEcompsocthanksitem  M. Luvisotto and M. Zorzi are with the Department of 
	Information Engineering, University of Padova, Via Gradenigo 6/B, 35131 Padova,
	Italy. E-mail: {luvisott, zorzi}@dei.unipd.it
	\IEEEcompsocthanksitem  F. Lahouti is with the Electrical Engineering 
	Department, California Institute of Technology, 1200 E California Blvd, MC 
	136-93, Pasadena, CA 91125. E-mail: lahouti@caltech.edu
	}
\thanks{Part of this work has appeared on the Proceedings of the 2016 IEEE 
	International Conference of Communications (ICC) \cite{ICC_Stat_QoS}.}}

\maketitle

\begin{abstract}
Statistical QoS provisioning as an important performance metric in analyzing next generation mobile cellular network, aka 5G, is investigated. In this context, by quantifying the performance in terms of the effective capacity, we introduce a lower bound for the system performance that facilitates an efficient analysis. Based on the proposed lower bound, which is mainly built on a per resource block analysis, we build a basic mathematical framework to analyze effective capacity in an ultra dense heterogeneous cellular network. We use our proposed scalable approach to give insights about the possible enhancements of the statistical QoS experienced by the end users if heterogeneous cellular networks  migrate from a conventional half duplex to an imperfect full duplex mode of operation. Numerical results and analysis are provided, where the network is modeled as a Mat\'ern point process. The results demonstrate the accuracy and computational efficiency of the proposed scheme, especially in large scale wireless systems. Moreover, the minimum level of self interference cancellation for the full duplex system to start outperforming its half duplex counterpart is investigated. 
\end{abstract}

\begin{IEEEkeywords}
Full Duplex, Heterogeneous Cellular Network, Statistical QoS, 5G.
\end{IEEEkeywords}

\vspace{+1 cm}

\IEEEraisesectionheading{\section{Introduction}\label{sec:introduction}}
\IEEEPARstart{T}{he} ever increasing demand for mobile data traffic continues with the advent of smart phones, tablets, mobile 
routers, and cellular M2M devices. This is accompanied by user
behavioral changes from web browsing towards video streaming, social 
networking, and online gaming with distinct quality of service (QoS) requirements \cite{ericsson2014}. To handle this evolution, researchers are examining different enabling technologies for 5G, including mmWave communications for wider bandwidth, extreme densification of the network via low power base stations, large-scale antenna systems (known as massive MIMO), and wireless full duplex communications~\cite{andrews20145G, InbandSabharwal}. 

The new cellular architecture known as Heterogeneous Cellular Networks (HCNs) 
refers to a scenario in which a macro cell is overlaid by 
heterogeneous low--power base stations (BSs). 
Such low power BSs have small coverage areas and are characterized by 
their own transmit power and named accordingly as micro, pico and femto cells. 
They are used to increase the capacity of the network while 
eliminating coverage holes \cite{HetNetfromTheorytoPractic}. In terms of modeling and analyzing HCNs, stochastic geometry has been vastly used in recent years \cite{StochGeoEk}. For instance Poisson point processes (PPPs) have provided tractable, simple, and accurate expressions for coverage probabilities and mean transmission rates \cite{StochGeoEk, DhillonkHetNet, AndrewsPPP}. However, in a real cellular network, the adoption of a simple PPP to model the location of the BSs does not capture an important characteristic of a real network deployment. Spectrum access policies, network planning, and the MAC layer impose constraints on the minimum distance between any two BSs or UEs in the network that are operating in the same resource blocks (RBs)~\cite{StochGeoEk}. According to these considerations, a repulsive point 
process with a minimum distance such as the Mat\'ern hard core point process 
(HCPP), despite its higher complexity, represents a better candidate to model a HCN 
compared to a simple PPP~\cite{StochGeoEk}, \cite{DengZH14}.

The Full Duplex (FD) radio technology can enhance spectrum 
efficiency by enabling simultaneous transmission and reception in the same 
frequency band at the same time. This new emerging technology has the potential 
to double the physical layer capacity \cite{InbandSabharwal} and enhance 
the performance even more, when higher layer protocols are redesigned accordingly~\cite{luvisotto2016rcfd}.

Due to the hurdles of canceling self--interference (SI) in FD devices via 
active and passive suppression mechanisms, FD operations are more reliable in 
low power wireless nodes. For instance in \cite{bharadia2013full} and 
\cite{bharadia2014full} the authors have implemented a FD WiFi radio operating 
in an unlicensed frequency band with 20~dBm transmit power while the same trend 
is followed in other works like \cite{DuarteExperimentalCharach}, where the 
maximum transmit power is 15~dBm. All these implementations suggest FD 
technology as a reasonable candidate to be employed in the low power BSs deployed within 
HCNs. 
Moreover, the increased spectral efficiency of the FD systems, combined with 
that of HCNs, provides another serious motivation in attempting to analyze a completely 
FD HCN. 

From another perspective, next generation mobile networks (5G) will aim not 
only to increase the network capacity but also to enhance several other 
performance metrics, including lower latency, seamless connectivity, and increased mobility \cite{andrews20145G}. These requirements can be generally 
regarded as improvements in the QoS experienced by 
the network entities. According to a forecast by Ericsson, 
approximately 55 percent of all the mobile data traffic in 2020 will account 
for mobile video traffic while another 15 percent will account for social 
networking~\cite{ericsson2014}. These multimedia services require a
bounded delay. The delay requirements of time sensitive services in 5G will greatly vary from milliseconds to a second~\cite{ZhangHetQoS}. 
Consequently, the analysis of statistical QoS in HCNs will become extremely important in 
the near future. 

The objective of this paper is to analyze and compare FD and HD HCNs 
in provisioning statistical QoS guarantees to the users in the network. In this regard, we model the HCN as a Mat\'ern HCPP, and statistically assess the QoS in terms of Effective Capacity (EC) or the maximum throughput under a delay constraint \cite{DapengQoS}. We provide insights on possible improvements in the QoS experience of end users if the current architecture 
migrates from conventional HD to FD. 
We propose a lower bound for the EC which greatly 
reduces the complexity of the analysis, while tightly approximating the system performance, 
especially in very large scale systems. The presented analysis is validated with simulation results.

The rest of the paper is organized as follows. A brief review of FD, 
statistical QoS provisioning, and stochastic geometry is provided in 
Section~\ref{Preliminaries}. The system model is 
described in Section~\ref{System Model}. 
Section~\ref{sec:theoretical} presents the proposed lower bound for the system 
performance and the corresponding theoretical analysis, whose results are validated 
through simulations in Section~\ref{sec:simulations}. Some discussion on user distributions in the network, coverage areas of the small cells, and channel models is provided in Section~\ref{Discussions}.  Finally, 
Section~\ref{sec:conclusions} concludes the paper.

\section{Preliminaries}
\label{Preliminaries}

\subsection{In-band Full Duplex Wireless Communications}
In--band full--duplex (IBFD) devices are capable of transmitting and receiving 
data in the same frequency band at the same time. In traditional wireless 
terminals, the self interference (SI) power is much larger than that of the received intended 
signal, making any reception infeasible while a transmission is 
ongoing. To overcome this issue, FD terminals are equipped with active and 
passive cancellation mechanisms to suppress their own SI in the received signal \cite{InbandSabharwal}. 
However, in practice, because of the many imperfections in transceiver characteristics and operations, namely phase noise, nonlinearities, I/Q imbalances, receiver quantization noise and constraints on the accuracy of the SI channel estimation, 
full cancellation of the SI signal is not possible. Therefore, some residual 
self--interference (RSI) always remains after all cancellation steps, and
results in a degraded system performance \cite{InbandSabharwal}, \cite{bharadia2013full}, \cite{bharadia2014full}, \cite{DuarteExperimentalCharach}. 

The RSI signal represents the main obstacle for a perfect FD communication and, 
similar to noise, is essentially uncorrelated with the original transmitted signal. Among different methods to model RSI, we consider RSI as a zero mean complex Gaussian random variable whose variance is a function of the transmitted power \cite{JSACRodriguez}, \cite{RamirezOptimal}. We have
\begin{equation}
\textbf{RSI} \sim \mathbb{CN} \left( {0,\eta {P^\kappa }} \right)
\label{eq}
\end{equation}
where $P$ is the transmit power, and $\eta$ and $\kappa$, $0 \le \eta,\kappa  \le 1$, are respectively the linear and the nonlinear SI cancellation parameters capturing the SI cancellation performance.
When no SI cancellation is performed $\eta , \kappa=1$, while $\eta=0$ 
represents the ideal case of perfect SI cancellation.

\subsection{Statistical QoS guarantees}
\label{subsec:qos_guarantees}

Real time multimedia services like video streaming require bounded delays. In 
this context, a received packet that violates its delay bound requirement is 
considered useless and is discarded. 
Due to the wireless nature of the access links in a mobile cellular network, 
providing deterministic delay bound guarantees is impossible. 
Thus, effective capacity defined as the maximum throughput under a given delay 
constraint, has been used to analyze multimedia wireless systems 
\cite{DapengQoS}. 
Any active entity in a cellular network can be regarded as a queueing system: 
it 
generates packets according to an arrival process, stores them in a queue and 
transmits them according to a service process.
For stationary and ergodic arrival and service processes, the 
queue length process $Q(t)$ converges, in distribution, to a random variable $Q(\infty)$ as follows
\begin{equation}
- \mathop {\lim }\limits_{{Q_{th}} \to \infty } \frac{{\log \left( {\Pr \left\{ {Q(\infty ) > {Q_{th}}} \right\}} \right)}}{{{Q_{th}}}} = \theta,
\label{QoS Definition}
\end{equation}
 where $Q_{th}$ is the queue length threshold and $\theta$ is a constant value called QoS exponent. In fact \eqref{QoS Definition} states that  the probability that the queue length process violates this threshold decays exponentially fast as the threshold increases, i.e., $\Pr \left\{ {Q(t) > Q_{th}} \right\} \sim {e^{ - \theta Q_{th}}}$ as $Q_{th} \rightarrow \infty $ \cite{Cheng}. In addition, it was shown in \cite{DapengQoS} that for small values of $Q_{th}$, the violation probability can be approximated more accurately as $\Pr \left\{ {Q(t) > Q_{th}} \right\} \approx \delta {e^{ - \theta Q_{th}}}$, where $\delta$ represents the probability that the buffer is not empty, and the QoS exponent, $\theta$, characterizes the decaying rate of the QoS violation probability. A large value for $\theta$ corresponds to a faster decaying rate, which means that the system has stringent statistical delay constraints. For instance, when $\theta \rightarrow \infty$ the system is very sensitive to delay and generally no delay is tolerable. On the other hand, for small values of $\theta$ the system can tolerate the delays and, when $\theta \rightarrow 0$ there is no sensitivity at all to delays in the system, and in this case the EC tends to the Shannon capacity \cite{ZhangQoS}.  

Let us define the service provided by the channel until time 
slot $t$ as 

\begin{equation}
C\left( {0,t} \right) = \sum\limits_{k = 1}^t {R\left[ k \right]}
\label{eq3}
\end{equation}
where $R[k]$ denotes the bits served in time slot $k$. 
The effective capacity of the channel is defined as \cite{DapengQoS}

\begin{equation}
{{EC}}\left( \theta  \right) =  - \frac{{{\Lambda _C}\left( { - \theta } \right)}}{\theta }
\label{eq4}
\end{equation}
where 
$\Lambda_C\left(-\theta\right)=\lim\limits_{t\to\infty}\frac{1}{t}\log E\left\{
e^{-\theta C(0,t)}\right\}$
is the G\"artner - Ellis limit of the service process $C(0,t)$. 
If the instantaneous service process, $R[k]$, is independent over time, 
EC can be simplified to 

\begin{equation}
{{EC}}\left( \theta  \right) =  - \frac{1}{\theta }\log {E \left \{ {{e^{ - 
				\theta R[k]}}} \right \}}.
\label{eq5}
\end{equation}

EC is defined as the maximum arrival rate a given service process can support in order to guarantee a QoS requirement specified by $\theta$ \cite{ZhangQoS}. We note that EC is defined for the service process and not for the arrival process. However, the dual definition of EC for the arrival process, referred to as the effective bandwidth, serves this purpose. It indicates the minimum service process required by a given arrival process for which the QoS exponent $\theta$ is satisfied \cite{ZhangQoS}.  

\section{System Model}
\label{System Model}
We consider a HCN consisting of a macro cell, overlaid by different tiers of small cells, each with its own 
characteristics including transmit power, path loss exponent, and coverage 
range. 
Each tier is assumed to have a circular coverage area provided by an omni-directional antenna to serve any user within its coverage range. The small cells are assumed to use out of band resources 
like fiber optics, wire, or microwave links for backhauling. Fig.~\ref{fig1} depicts an example of such a HCN.  
When the system is HD, a conventional frequency-division-duplexed (FDD) HCN  is assumed, which separates up-link (UL) and down-link (DL) transmissions, while in the FD case we 
consider a completely FD HCN where all the network entities are assumed to be 
(imperfect) FD devices which communicate in bidirectional FD mode, as depicted in Fig. \ref{fig1}. 

To model the location of small cells and the distribution of the user equipments (UEs) within each small cell we use a Mat\'ern HCPP. In this process the locations of BSs follow a PPP with a hard core distance constraint between BSs. The locations of the users associated to 
each BS are assumed to be uniformly distributed in the coverage area of the corresponding 
BS. Moreover, in line with common practice, we assume that the positions of the BSs are known by the 
network operator, leading to a Mat\'ern HCPP with known locations for the cluster 
heads (BSs).


\begin{figure}
	\centering
	\scalebox{0.9}{\begin{tikzpicture}
		
		\node at (0,1.15) {\includegraphics[scale=0.3]{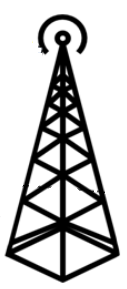}};
		\node at (1.35,-0.2) {\includegraphics[scale=0.06]{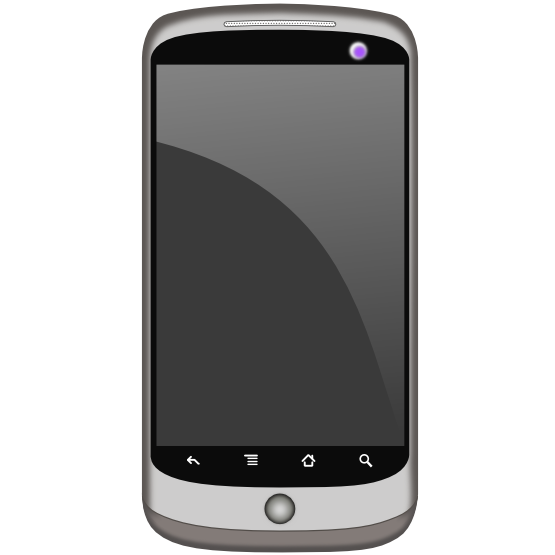}};
		\node at (3.75,-.35) {\includegraphics[scale=0.06]{ue.png}};
		\draw[draw=black,dotted] (0,0) ellipse (4.3cm and 1.95cm);
		
		\draw[draw=black,dashed] (-2.5,-0.75) ellipse (1.25cm and .45cm);
		\draw[draw=black,dashed] (-1.3,+1.25) ellipse (.85cm and .2cm);
		\draw[draw=black,dashed] (+2.25,+0.75) ellipse (1cm and .3cm);
		
		\node at (-2.5,-0.2) {\includegraphics[scale=0.15]{macro_cell.png}};
		\node at (-1.3,+1.7)  {\includegraphics[scale=0.11]{macro_cell.png}};
		\node at (+2.25,1.25) {\includegraphics[scale=0.13]{macro_cell.png}};
		
		\draw[draw=black,dashed] (-2.65,+.7) ellipse (0.48cm and .15cm);
		\draw[draw=black,dashed] (+1.25,+1.35) ellipse (0.45cm and .095cm);
		\draw[draw=black,dashed] (0.25,-1.1) ellipse (0.44cm and .11cm);
		\draw[draw=black,dashed] (1,-1.5) ellipse (0.47cm and .14cm);
		
		\node at (-2.64,1) {\includegraphics[scale=0.075]{macro_cell.png}};
		\node at (+1.2,+1.65) {\includegraphics[scale=0.07]{macro_cell.png}};
		\node at (0.23,-.78) {\includegraphics[scale=0.08]{macro_cell.png}};
		\node at (1,-1.15) {\includegraphics[scale=0.08]{macro_cell.png}};
		
		\draw[draw=black,dashed] (2.75,-.45) ellipse (1.3cm and 0.4cm);
		
		\node at (2.65,0) {\includegraphics[scale=0.12]{macro_cell.png}};
		
		\draw[draw=black,dashed] (-1.3,+1.25) ellipse (.85cm and .2cm);
		\draw[draw=black,dashed] (+2.25,+0.75) ellipse (1cm and .3cm);	
		
		\node at (-0.2,-.45) {\scriptsize Macro};
		\node at (-3,-1) {\scriptsize Pico};
		\node at (-1.8,+1.25) {\scriptsize Pico};
		\node at (-1.8,+1.25) {\scriptsize Pico};
		\node at (+1.7,.7) {\scriptsize Pico};
		\node at (3,-1.1) {\scriptsize Picocell};
		\node at (-3.1,0.45) {\scriptsize Femto};
		\node at (.25,-1.35) {\scriptsize Femto};
		\node at (1.05,-1.76) {\scriptsize Femto};
		\node at (1.75,1.55) {\scriptsize Femto};

		\draw [latex-latex,draw=black] (2.75,0.35)--(3.55,-.2);
		\node at (-3.35,-.65) {\includegraphics[scale=0.06]{ue.png}};
		\node at (-1.65,-.75) {\includegraphics[scale=0.06]{ue.png}};
		
		\draw [latex-latex,draw=black] (-3.25,-0.35)--(-2.65,+.35);
		\draw [latex-latex,draw=black,dashed] (-2.35,+.35)--(-1.65,-.5);
		\draw [latex-latex,draw=black,dotted] (.2,2.1)--(1.25,0.15);
		\end{tikzpicture}}
	\caption{System model: an HCN with one macrocell and several low-power BSs.}
	\label{fig1}
\end{figure}


We consider a Rayleigh fading, path loss dominated, AWGN channel model.
As a result, the interference at the desired UE from another network entity located at distance $x$ is given by $Ph{\left\| x 
	\right\|^{- \alpha }}$, where $P$ is the transmit power of the interferer, $h$ is an 
exponential random variable modeling Rayleigh fading ($h \sim \exp \left( 1 
\right)$), and $\alpha$ represents the path loss exponent. 

In a FD HCN, a UE in the network experiences three different types of 
interference: (1) RSI, due to concurrent transmission and reception in the same 
frequency band at the same time and imperfect cancellation; (2) interference from BSs that are 
transmitting in the same resource blocks (RBs) as those in which the UE is being served; and 
(3) interference from other UEs in the network that are transmitting in the 
same RBs in which the UE is being served. In the HD scenario, we instead assume frequency--division duplexing, where the UE does not experience RSI or interference from other UEs. However, the interference from other BSs is still present.


\begin{figure}[!t]
	\centering
	\includegraphics[width=0.45 \textwidth]{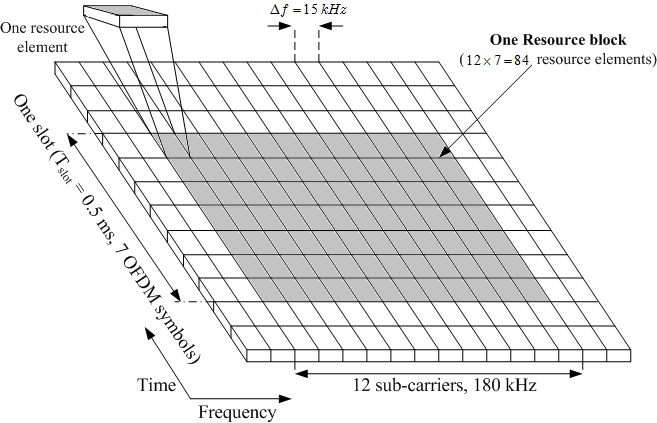}
	\caption{Illustration of one resource block in LTE-A.}
	\label{Fig1.5}
\end{figure}

In the FD HCN, the Signal--to--Interference plus 
Noise Ratio (SINR) at the desired FD UE is expressed as 
\begin{align} 
\nonumber
& \text{SINR}_{\text{FD}} = \\ \nonumber & \frac{{{P_{i}}{h_{{x_i}}}{\left\|x_i\right\|}{^{ - \alpha_i 
			}}}}{{\sum\limits_k \left( {{\sum\limits_{x \in {\Phi _k^{\text{BS}}}} 
			{{P_k}{h_x}\left\|x\right\|{^{ - {\alpha _k}}}} }  + 
		{\sum\limits_{y \in {\Phi _k^{\text{UE}}}} {{P_{\text{UE}}}} } {h_y}\left\|y\right\|{^{ - {\alpha 
					_k}}}} \right) + \eta {P^\kappa } + {\sigma ^2}}}.\\ 
\label{eq7}
\end{align}
In this notation, the numerator represents 
the desired signal power received from a BS in the $i^{th}$ tier which serves 
the UE. Here, a tier is defined as the set of BSs with a given density and the same 
characteristics including average transmit power, supported data rate, and
coverage area \cite{DhillonkHetNet}. 
The first and second terms in the denominator represent the interference from other 
BSs and UEs in the network operating in the same RBs as the desired UE. Specifically, $\Phi_k^{\text{BS}}$ and $\Phi_k^{\text{UE}}$ indicate sets containing the 
positions of all \textit{interfering} BSs and UEs in the $k^{\text{th}}$ tier, and the summation is over all 
possible tiers. The third term is the RSI signal power as modeled 
in \eqref{eq}. Finally, $\sigma^{2}$ is the additive noise power. 

This paper analyzes the system performance on a resource block (RB) basis. An illustration of a single RB in LTE-A is presented in Fig. \ref{Fig1.5}. The reason for following a RB analysis is three-fold: i) the service rate, $R[k]$, can be replaced by the instantaneous capacity of the channel. This is due to the fact that the service time, $T_f$, is less than the coherence time of the channel, but still sufficiently long to support the information theoretic assumption of adaptive modulation and capacity achieving coding; ii) the service rate, $R[k]$, becomes statistically independent from one RB to another, as was assumed in deriving \eqref{eq5}; iii) most importantly, the RB based analysis will assist to find a very tight, computationally efficient and scalable lower bound for analyzing EC in the meaningful range of values for the QoS exponent. Consequently, in this work, the unit for both the service rate,
$R[k]$, and the effective capacity, $EC(\theta)$, is defined as bits per RB. 

We recall that the number of bits delivered to a UE during an interval, $T_f$, in a given 
bandwidth, $BW$, if capacity achieving modulation and coding are used, can be 
represented as 

\begin{equation}
\text{R}=T_f BW\log_2\left(1+\text{SINR}\right)\text{.}
\label{eq8}
\end{equation}
Therefore, the effective capacity of the desired UE based on \eqref{eq5} can be 
expressed as 
\begin{align}
\text{EC}\left(\theta\right)&= 
-\frac{1}{\theta}\log\left(\mathbb{E}\left\{\text{exp}\left(-\theta T_f 
BW\log_2\left(1+\text{SINR}\right)\right)\right\}\right) \nonumber\\
\label{eq9}
& = 
-\frac{1}{\theta}\log\left(\mathbb{E}\left(\left(1+\text{SINR}\right)^{-\theta 
	T_f BW\log_2e}\right)\right)
\end{align}
where the expectation is taken with respect to the SINR. 

In a HD scenario, a 1/2 scaling factor is needed and, also, the SINR would become
\begin{align} 
\text{SINR}_{\text{HD}} = \frac{{{P_i}{h_{{x_i}}}{\left\|x_i\right\|}{^{ - \alpha_i 
			}}}}{{\sum\limits_k {\sum\limits_{x \in {\Phi _k^{\text{BS}}}} 
			{{P_k}{h_x}\left\|x\right\|{^{ - {\alpha _k}}}} } + {\sigma ^2}}}
\label{Dual-eq7}
\end{align}
where the numerator represents 
the desired signal power received from a BS in the $i^{th}$ tier, which serves 
the UE and the denominator includes noise and DL interference from interfering  BSs in the $k^{th}$ tier and the summation is over all tiers.

\section{Theoretical analysis}
\label{sec:theoretical}
We aim at computing the QoS experienced by a generic UE, that can be placed 
anywhere in the coverage area of its own small cell with uniform distribution. 
To find the exact EC in a given topology, one needs to compute 
\eqref{eq9} either through extensive simulations or by mathematical analysis. 
It is worth mentioning that, if there are $M$ small cells within the macro 
cell communicating with $M$ FD UEs, the associated integrals (or simulation setups) would be in a $2M+1$ 
dimensional parameter space. To understand why this is the case just note that the expectation in \eqref{eq9} is with respect to SINR which is a function of desired and interferer's signal strength. This is seen in \eqref{eq9} as the expectation with respect to SINR is a function of both desired signal and interference.

\subsection{Approximating EC}
\label{subsection:approximating the EC}

Let us define a generic function $g$ of $s, I, a,$ and  $\beta$ as 
follows

\begin{equation}
g(s,I) \buildrel \Delta \over =  \left(1+\frac{s}{I+a}\right)^{-\beta}.
\label{eq10}
\end{equation}
This function has the same structure of the expectation argument in 
\eqref{eq9}, where $s$ models the received signal power, $I$ represents the 
overall interference from other BSs and UEs in the network, $a$ represents the 
RSI and noise, and $\beta=\theta\cdot T_f\cdot BW\cdot\log_2 e$. Based on this 
definition, the EC can be rewritten as

\begin{equation}
EC(\theta)=-\frac{1}{\theta}\log \left( \E\limits_{s,I}  g \left( s,I  
\right)\right)
\label{eq11}
\end{equation}

\begin{Lemma}
	\label{Lemma1}
	For $0\le \beta \le 1$, $g$ is always a concave function of $I$. 
\end{Lemma}

\begin{proof}
	By assuming $s$ as a constant, taking the second derivative of $g$ with respect 
	to $I$ leads to 
	\begin{align}
	\nonumber
	& \frac{{{\partial ^2}g(s,I)}}{{{\partial}{I^2}}} = \\ & \frac{{\beta s}}{{{{\left( {I 
						+ a} \right)}^4}}}{\left( {1 + \frac{s}{{I + a}}} \right)^{ - \left( {\beta  + 
				2} \right)}}\left( { - 2\left( {I + a} \right) + \left( {\beta  - 1} \right)s} 
	\right)
	\label{eq12}
	\end{align}
	which is negative, for all values of $0\le \beta <1 + \frac{2}{{\text{SINR}}}$, where SINR$=\frac{s}{I+a} \ge 0$. This readily shows that $g$ is a concave function of $I$ for any $0\le \beta \le1$. 
\end{proof}
The concavity of $g$ helps us find a tight lower bound for the EC with greatly 
decreased complexity. To this end, by exploiting Jensen's inequality, we obtain
\begin{equation}
EC_{\text{LB}}\left( \theta  \right) =\frac{{ - 1}}{\theta }\log {\E_s}\left[ 
{{{\left( {1 + \frac{s}{{\bar{I}  + a}}} \right)}^{ - \beta }}} \right] 
\le EC\left( \theta  \right) 
\label{eq13}
\end{equation} 
Here, $\bar{I}=\E_I(I)$ is the average interference experienced by the UE, and 
the remaining expectation only applies to the desired signal power. The 
advantage of this lower bound is its much smaller computational complexity. Indeed, 
calculating this lower bound only requires a 1--dimensional integral (or 
simulation) with respect to the desired signal power provided that the average interference on the desired network entity ($\bar I$) is known. 
Therefore, the proposed bound makes this calculation scalable with the size of 
the network, at the possible cost of losing some precision. 

In order to efficiently compute the proposed lower bound in \eqref{eq13}, i.e., $EC_{\text{LB}}\left(\theta\right)$, one has to calculate analytically the 
average interference on the desired UE, $\bar {I} $. Recalling the expression of the interference from 
Section~\ref{System Model}, the average power from an interferer located at distance $x$ from the considered UE can be found as

\begin{equation}
	\E[Ph\Vert x\Vert^{-\alpha}] \mathop  = \limits^{\rm{(a)}}  P\E (h) \E[\Vert 
	x\Vert^{-\alpha}] \mathop  = \limits^{\rm{(b)}} P \E[\Vert x\Vert^{-\alpha}]
	\label{eq14}
\end{equation}
where (a) follows from the fact that the channel coefficient and distance 
between the interferer and the desired UE are independent random variables and 
(b) holds because we assumed the Rayleigh fading channel, $h$, has unit mean. Consequently, all 
our efforts will be focused on finding the average path loss from the desired 
UE to the interferers. 

It is noteworthy that in deriving the lower bound proposed in \eqref{eq13} we implicitly have a constraint on the QoS exponent, $\theta$, through Lemma 1, where we assumed the constraint on $\beta$ as $0 \le \beta \le 1$. This imposes a constraint on $\theta$, as $0 \le \theta \le \frac{1}{T_f BW \log_2 e}$, while in general, the QoS exponent is defined in the range $[0,\, 
\infty)$. Due to the per RB based analysis this constraint on $\theta$ would amount to 
\begin{equation}
0 \le \theta \le \frac{1}{T_f BW \log_2 e} \approx 10^{-2}.
\label{constraintQoS}
\end{equation}
Fortunately, the range of $\theta$ in \eqref{constraintQoS} lies in the meaningful range of the QoS exponent \cite{ZhangQoS}. However, simulations show that our method  gives a good approximation for EC even in a wider range.

\begin{figure}[!t]
	\centering
	\begin{tikzpicture}
	\draw[draw=black,thick,dotted] (-3,-.7) circle (.95);
	\draw[draw=black,thick,dotted] (0,-.1) circle (1.3);
	\draw[draw=black,fill=black] (-1.55,+1.15) circle (0.066);
	\draw[draw=black,fill=black] (0,-.1) circle (0.06);
	\draw[draw=black,fill=black] (-3,-.7) circle (0.06);
	\draw[draw=black,thick,dashed] (-5.5,+1.5) arc (183:365:3.85cm);
	\draw[draw=black,thick] (-2.81,-0.65) arc (10:130:0.2cm);
	\draw[draw=black,thick] (0.25,-.04) arc (-10:95:0.2cm);
	\draw[draw=black,thick] (0.01,.06) arc (90:170:0.2cm);
	\draw[draw=black,thick] (-.8,-.12) arc (160:210:0.15cm);
	
	\draw[latex-,draw=black,thick,dashed] (-5.45,+1.2) -- (-1.55,+1.15);	
	\draw[-latex,draw=black,dashed] (-3,-.7) -- (-3,-1.65);	
	\draw[-latex,draw=black,dashed] (0,-.1) -- (0.25,-1.355);
	\draw[line width=0.7pt, draw=black] (-3,-.7) -- (0.5,.02);
	\draw[-latex, draw=black] (-3,-.7) -- (-3.4,-.2);
	\draw[-latex, draw=black] (0,-.1) -- (0.1,0.65);
	\draw[draw=black] (0.1,0.65) -- (-3.4,-.2);
	\draw[draw=black] (-3.4,-.2) -- (0,-.1);
	
	\node[label=below:\rotatebox{0}{$R_1$}] at (-3.15,-.8) {};
	\node[label=below:\rotatebox{0}{$R_2$}] at (+0.3,-.405) {};
	\node[label=below:\rotatebox{0}{$r_1$}] at (-3.35,-.25) {};
	\node[label=below:\rotatebox{0}{$\theta_1$}] at (-2.8,0) {}; 
	\node[label=below:\rotatebox{0}{$r_2$}] at (-.11,0.7) {};
	\node[label=below:\rotatebox{0}{$\theta_2$}] at (0.41,0.6) {};  
	\node[label=below:\rotatebox{0}{$d$}] at (-1.75,-0.23) {};
	\node[label=below:\rotatebox{0}{$c$}] at (-1.75,.28) {};
	\node[label=below:\rotatebox{0}{$x$}] at (-1.75,0.65) {};
	\node[label=below:\rotatebox{0}{$\gamma$}] at (-.25,0.35) {};
	\node[label=below:\rotatebox{0}{$\psi$}] at (-1.1,0.2) {};
	
	\node at (-1.55,+1.35) {Macro BSs};
	
	\end{tikzpicture}
	\caption{Structure of the interference on the desired UE.}
	\label{fig2}
\end{figure}
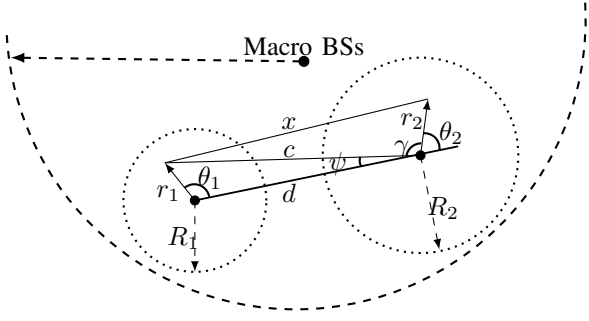

\subsection{Average Interference on a UE}
As mentioned in the previous section we need to provide a framework to find $\bar I$. In the rest of this section, we present the building blocks to achieve this goal. 
\label{subsec:average_interference}

Fig. \ref{fig2} depicts a deployment of two small cells and their 
corresponding coverage areas within a macro cell. This is an example of a 
Mat\'ern HCPP with two cluster heads (BSs) and a hard core distance $ r_h \ge 
{R_1+R_2} $. This assumption for the hard core distance makes the two small 
cells non--overlapping. The probability density functions (PDFs) of the 
interferer and desired UE locations, expressed in polar coordinates, are 
	\begin{equation}
	f_i\left( {{r_i},{\theta_i}} \right) = \left\{ \begin{array}{l}
	\frac{{{r_i}}}{{{ {\pi} } {{R_i}^2}}}\,\,\,\,\,\,\,\,\,\,\,\,\,\,\,\,\,\,\,\,\,\,\,{0 < {r_i} < R_i, \quad 0 < \theta_i < 2 \pi }\,\,\,\,\,\,\,\,\,\,\,\,\,\,\\
	0\,\,\,\,\,\,\,\,\,\,\,\,\,\,\,\,\,\,\,\,\,\,\,\,\,\,\,\,\,\,\,\,\,\,\,\,\,\, \textrm{Otherwise}
	\end{array} \right.
	\label{UEDis}
	\end{equation}
	respectively for  $i= 1,2$; see Fig.~3.

\begin{proposition}
	\label{proposition1}
	Consider a UE being served in a small cell in presence of an interferer BS, according to Fig.~3. The average interference from a specific interferer BS on the desired UE is 
	\begin{equation}
	{I_{{\rm{BS - UE}}}}=P_{BS}{d^{ - \alpha }}\left[ {1 + \frac{{{{\alpha 
						^2}}}}{8}\left( {\frac{{{R_2^4}}}{{3{d^4}}} + \frac{{{R_2^2}}}{{{d^2}}}} 
		\right)+\frac{{ \alpha }}{4} {\frac{{{R_2^4}}}{{3{d^4}}}}} \right].
	\label{eq18}
	\end{equation}
	where $P_{BS}$ is the transmit power of the BS, $d$ is the distance between the interferer BS and the BS serving the desired UE, $\alpha$ is the pathloss exponent, and $R_2$ is the coverage radius of the small cell in which the desired UE is being served. 		
\end{proposition} 

\begin{proof}
	\label{proposition1proof}

We can find the squared distance between the desired UE and the interferer as (Fig.~3)
\begin{equation}
{x^2} = {c^2} + r_2^2 - 2c{r_2}\cos \left( \gamma  \right)
\label{eq15}
\end{equation}
where 
\begin{align}
{c^2} =& r_1^2 + {d^2} - 2{r_1}d\cos \left( {{\theta _1}} \right) \\
&\gamma \,\, = \pi - \theta _2 - \psi
\end{align}
and $\psi$ is also a random variable 
depending on the interferer's position.

Our goal is to compute the average path loss between the considered UE and the interferers
\begin{equation} 
\E\left[ {{{\left\| x \right\|}^{ - \alpha }}} \right] = \E \left[ {{{\left( 
			{{{ {{c^2} + r_2^2 + 2c{r_2}\cos \left( {{\theta _2} + \psi } \right)} }}} 
			\right)}^{ - \frac{\alpha }{2}}}} \right]\
\label{eq16}
\end{equation}
which is challenging to compute in general. 
To this end, we first compute the expectation in \eqref{eq16} by assuming a 
fixed position for the interferer, i.e., fixed $(r_1,\theta_1)$. Subsequently, 
we compute the expectation of the resulting quantity with respect to all possible 
values of $(r_1,\theta_1)$. 

Regarding the first step, since we assumed $(r_1,\theta_1)$ is fixed, 
$c$ and $\psi$ become constants, which facilitates the analysis

\begin{align}
&\E\limits_{\left( {{r_2},{\theta _2}} \right)} \left[ {\left. {{{\left( 
				{{c^2} + r_2^2 + 2c{r_2}\cos \left( {{\theta _2} + \psi } \right)} \right)}^{ - 
				\frac{\alpha }{2}}}} \right|\left( {{r_1},{\theta _1}} \right)} \right]\\ 
&= \int\limits_0^{2\pi } {\int\limits_0^{R_2} {{{\left( {{c^2} + r_2^2 + 
					2c{r_2}\cos \left( {{\theta _2} + \psi } \right)} \right)}^{ - \frac{\alpha 
				}{2}}}} } .\frac{1}{\pi }\frac{{{r_2}}}{{{R_2^2}}}{\mkern 1mu} {\mkern 1mu} 
dr_2d\theta_2 \\ \nonumber
&= \int\limits_0^{2\pi } {\int\limits_0^{R_2} {{c^{ - \alpha }}{{\left( {1 + 
					{{\left( {\frac{{{r_2}}}{c}} \right)}^2} + 2\left( {\frac{{{r_2}}}{c}} 
					\right)\cos \left( {{\theta _2} + \psi } \right)} \right)}^{ - \frac{\alpha 
				}{2}}}} } \frac{1}{\pi }\frac{{{r_2}}}{{{R_2^2}}}{\mkern 1mu} {\mkern 1mu} 
d{r_2}d{\theta _2} \\
&\mathop  \simeq \limits^{(a)} {c^{ - \alpha }}\int\limits_0^{2\pi } 
{\int\limits_0^{R_2} 
	{\left\{ {1 - \frac{\alpha }{2}\left[ {{{\left( {\frac{{{r_2}}}{c}} \right)}^2} 
				+ 2\left( {\frac{{{r_2}}}{c}} \right)\cos \left( {{\theta _2} + \psi } 
				\right)} \right]} \right.  } } +\\ \nonumber
&\left. {\frac{\alpha \left( {\alpha  + 2} \right)}{8}{{\left[ {{{\left( 
						{\frac{{{r_2}}}{c}} \right)}^2} + 2\left( {\frac{{{r_2}}}{c}} \right)\cos 
				\left( {{\theta _2} + \psi } \right)} \right]}^2}} \right\}.\frac{1}{\pi 
}\frac{{{r_2}}}{{{R_2^2}}}{\mkern 1mu} {\mkern 1mu} d{r_2}d{\theta _2}\\ 
\label{eq17} 
& = {c^{ - \alpha }}\left[ {1 + \frac{{{{\alpha ^2}}}}{8}\left( 
	{\frac{{{R_2^4}}}{{3{c^4}}} + \frac{{{R_2^2}}}{{{c^2}}}} \right)+\frac{{ \alpha 
		}}{4} {\frac{{{R_2^4}}}{{3{c^4}}}}} \right] \\ \nonumber
	\label{fixedr1theta1}
	\end{align}
	where in (a) we used the first three terms of the Taylor series expansion of 
	${\left( {1 + x} \right)^{ - \omega}} = 1 - \omega x + \frac{{\omega\left( 
			{\omega + 1}\right)}}{2!}{x^2} +\dots $. 
	The Taylor approximation is legitimate if $x<1$ $(c>R_2)$ which is already 
	satisfied considering the repulsive point process we have assumed for the small 
	cells, characterized by the hard core distance $r_h\ge{R_1+R_2}$.
	
	We recall that this result holds for any fixed values of 
	$\left(r_1,\theta_1\right)$. In particular, by setting $r_1 \rightarrow 0$ 
	(i.e., $c \rightarrow d$ in \eqref{eq17}), we obtain the average path loss 
	component between a randomly deployed UE and the BS. Therefore, the average 
	interference that an external BS causes to the considered UE uniformly placed 
	in any point within the coverage area of its small cell would become
	\begin{equation}
	{I_{{\rm{BS - UE}}}}=P_{BS}{d^{ - \alpha }}\left[ {1 + \frac{{{{\alpha 
						^2}}}}{8}\left( {\frac{{{R_2^4}}}{{3{d^4}}} + \frac{{{R_2^2}}}{{{d^2}}}} 
		\right)+\frac{{ \alpha }}{4} {\frac{{{R_2^4}}}{{3{d^4}}}}} \right].
	\label{eq18}
	\end{equation}
\end{proof}	
Due to the possible concurrent transmissions on the UL and DL in the FD scenarios, in the next proposition the average interference from another UE is provided.

\begin{proposition}
	\label{proposition2}
	The average interference from a uniformly deployed UE located in a neighboring small cell on the desired UE~is
	\begin{align}
		\nonumber
		&\,\,{I_{{\text{UE - UE}}}} = {P_{UE}} \cdot \left( {\mathop {\E} \limits_{\left( {{r_1},{\theta _1}} \right)} \left[ {{c^{ 
						- \alpha }}} \right]} + \right. \\
		&\left. {{ \frac{{{\alpha ^2}{R_2^2}}}{8}\mathop {\E} 
				\limits_{\left( {{r_1},{\theta _1}} \right)} \left[ {{c^{ - (\alpha  + 2)}}} 
				\right] + \frac{{\alpha \left( {\alpha  + 2} \right){R_2^4}}}{{24}}\mathop {\E} 
				\limits_{\left( {{r_1},{\theta _1}} \right)} \left[ {{c^{ - (\alpha  + 4)}}} 
				\right]}} \right).
		\label{eq19}
	\end{align}
where $P_{UE}$ is the transmit power of the interferer UE, and $c$ is the distance between the interferer and the BS which serves the UE. The expectation terms are with respect to the location of the interferer, i.e., $(r_1,\theta_1)$. 
\end{proposition}
\begin{proof}
	\label{proposition2proof}
	We compute the expectation of \eqref{eq17} with respect to the position of the interferer, i.e., $(r_1,\theta_1)$, which leads to \eqref{eq19}. 
\end{proof}	
	It must be noted that to compute the quantity in \eqref{eq19}, one needs to calculate only 
	$\E\limits_{\left({{r_1},{\theta _1}} \right)}\left[c^{-\alpha}\right]$, since 
	the other two expectations can be immediately obtained by replacing $\alpha$ 
	with $\alpha+2$ and $\alpha+4$.
	We further observe that this expectation corresponds to the average path loss 
	component between the interferer and the desired UE's BS. This quantity can be derived from \eqref{eq18} by setting $P_{BS}=1$ and changing $R_2$ to 
	$R_1$.
	
	The proposed relations in \eqref{eq18} and \eqref{eq19} can hence be used as a 
	basic mathematical tool to investigate the system performance. 

	\subsection{Average Interference from a free UE}
	\label{Average Interference on a UE from a free UE} 
	We should note that the only source of interference on a desired UE that 
	does not follow the structure provided in Fig. \ref{fig2} is the UE 
	connected to the macro BS{, that we refer to as \textit{free } UE. We 
	make here a simplifying assumption in order to derive an approximate 
	expression for the interference caused by the free UE to the desired UE, 
	of which we are computing the effective capacity.
	Specifically, let $R_1'$ be the coverage radius of the macro BS, $R_2'$ the 
	coverage radius of the small cell where the desired UE is placed, $d'$ the 
	distance between the macro BS and the BS of the small cell and 
	$\theta^{\ast}$ the angle of the circular sector centered in the macro BS 
	which fully encloses the coverage area of the small cell, as represented 
	in Fig.~\ref{fig2.5}. We assume that the free UE can be deployed anywhere 
	in the coverage area of the macro BS except in the sector with angle 
	$2\theta^{\ast}$, i.e., the shaded area in Fig.~\ref{fig2.5}}. Thus, the 
	PDF of the free UE in the polar coordinates of the macro cell would be
	\begin{equation}
	f\left( {{r^\prime_1},{\theta^\prime_1}} \right) = \left\{ \begin{array}{l}
	\frac{{{r^\prime_1}}}{{{\left( {\pi  - {\theta ^ * }} \right)} {{R'_1}^2}}}\,\,\,\,\,\,\,\,\,\,\,\,\,\,\,\,\,\,\,\,\,\,\,{\theta ^ * } < {\theta^\prime_1} < 2\pi  - {\theta ^ * }\,\,\,\,\,\,\,\,\,\,\,\,\,\,\\
	0\,\,\,\,\,\,\,\,\,\,\,\,\,\,\,\,\,\,\,\,\,\,\,\,\,\,\,\,\,\,\,\,\,\,\,\,\,\, -{\theta ^*} < {\theta^\prime_1} < {\theta ^ * }
	\end{array} \right.
	\label{freeUEDis}
	\end{equation}
The next proposition presents the average interference from a free UE in this setting. 
	\begin{proposition}
		\label{proposition3}
	The \textit{approximate} average interference from the free UE on a desired UE can be found by calculating
	\begin{align}
	\nonumber
	& {I_{{\text{free UE - UE}}}} \approx {P_{{\rm{UE}}}}. \\ \nonumber
	&\left[ {\left( {\frac{{\left( {\alpha  + 2} \right)\left( {\alpha  + 3} \right)}}{6} + \frac{2}{{\alpha  - 2}}} \right)} \right.\frac{1}{{{{R'_1}^2}}}\E\left( {{{c'}^{ - (\alpha  - 2)}}} \right) \\ \nonumber
	&- \frac{{\alpha \left( 3 \alpha^3 +11 \alpha^2-18 \alpha -56 \right)}}{{15(\pi  - {\theta ^*})\left( {{\alpha ^2} - 1} \right)}}\frac{{\sin {\theta ^*}}}{{{{R'_1}^2}}}\E\left( {{{c'}^{ - (\alpha  - 2)}}\cos \psi '} \right) \\ 
	&\left. { - \frac{{\left( {\alpha  + 2} \right)\left( {\alpha  + 4} \right)}}{{16(\pi  - {\theta ^*})}}\frac{{\sin 2{\theta ^*}}}{{{{R'_1}^2}}}\E\left( {{{c'}^{ - (\alpha  - 2)}}\cos 2\psi '} \right)} \right].
	\label{eq19.5}
	\end{align}
where $c'$ is the distance between the macro BS to the desired UE which may deployed in any location of the macro cell according to \eqref{freeUEDis}. Moreover $\E\left( {{{c'}^{ - (\alpha  - 2)}}\cos \psi '} \right)$, and $\E\left( {{{c'}^{ - (\alpha  - 2)}}\cos 2\psi '} \right)$ are provided in \eqref{Ec(cospsi)} and \eqref{Ec(cos2psi)}, respectively.
\end{proposition}
The proof of this proposition is provided in Appendix \ref{AppB}. It is worth noting that in the simulations, we let the free UE be in any location of the macro cell except the small cell under investigation. We will then draw comparisons between simulations and the analysis.  

  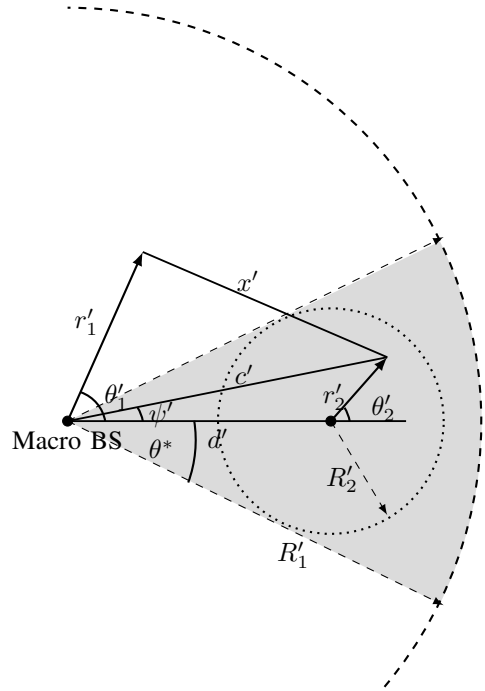
\begin{figure}[!t]
  	\centering
  	\begin{tikzpicture}
  	\draw[draw=none,fill=black!20!white,opacity=0.7] (-5,1) -- (-0.05,-1.43) 
  	arc (334:386:5.5cm) -- (-5,1);
  	\draw[draw=black,thick,dotted] (-1.5,1) circle (1.5);
  	\draw[draw=black,fill=black] (-5,1) circle (0.066);
  	\draw[draw=black,fill=black] (-1.5,1) circle (0.066);
  	\draw[line width=0.7pt, draw=black] (-5,1) -- (-.5,1);
  	\draw[-latex,draw=black,thick] (-1.5,1) -- (-.75,1.85);	
  	\draw[-latex,draw=black,thick] (-5,1) -- (-4,3.25);
  	\draw[line width=0.7pt, draw=black] (-4,3.25) -- (-.75,1.85);
  	\draw[line width=0.7pt, draw=black] (-5,1) -- (-.75,1.85);
  	\draw[-latex, draw=black,dashed] (-1.5,1) -- (-.75,-.25);
  	\draw[-latex, draw=black,dashed] (-5,1) -- (0,-1.45);
  	\draw[-latex, draw=black,dashed] (-5,1) -- (0,3.45);
  	
  	\draw[draw=black,thick,dashed] ([shift=(320:5.5cm)]-5,1) arc 
  	(320:450:5.5cm);
  	\draw[draw=black,thick] (-4.5,1) arc (10:70:.5cm);
  	\draw[draw=black,thick] (-4,1) arc (10:32:.5cm);
  	\draw[draw=black,thick] (-1.25,1) arc (10:32:.5cm);
  	\draw[draw=black,thick] ([shift=(333:1.8cm)]-5,1) arc (340:369:1.65cm);

  	\node[label=below:\rotatebox{0}{${R^{\prime}_2}$}] at (-1.35,.65) {};
  	\node[label=below:\rotatebox{0}{$R^{\prime}_1$}] at (-2,-.35) {};
  	\node[label=below:\rotatebox{0}{$r^{\prime}_1$}] at (-4.75,2.75) {};
  	\node[label=below:\rotatebox{0}{$\theta^{\prime}_1$}] at (-4.35,1.75) {}; 
  	\node[label=below:\rotatebox{0}{$\psi^{\prime}$}] at (-3.75,1.5) {};
  	\node[label=below:\rotatebox{0}{$r^{\prime}_2$}] at (-1.45,+1.75) {};
  	\node[label=below:\rotatebox{0}{$\theta^{\prime}_2$}] at (-.8,1.6) {};  
  	\node[label=below:\rotatebox{0}{$d^{\prime}$}] at (-3,1.15) {};
  	\node[label=below:\rotatebox{0}{$c^{\prime}$}] at (-2.65,2.) {};
  	\node[label=below:\rotatebox{0}{$x^{\prime}$}] at (-2.6,3.25) {};
  	\node[label=below:\rotatebox{0}{$\theta^{\ast}$}] at (-3.75,1) {};

  	\node at (-5,.75) {Macro BS};
  	
  	\end{tikzpicture}
  	\caption{Structure of interference caused by a free UE on the desired UE. 
  	{The shaded sector is the area where we assume that the free UE can 
  	not be deployed.}}
  	\label{fig2.5}
  \end{figure}

\section{Simulations and Results}
\label{sec:simulations}
{
This section provides numerical simulations to analyze the performance of the FD and HD HCNs under statistical QoS constraints. In addition, results are reported to assess the scalability of our proposed approach to analyze system performance. Although a single macro cell overlaid with randomly placed small cells is considered, the analytic framework proposed in this paper readily allows the study of a scenario with multiple macro cells. In the reported results, the UE in the dashed small cell is under consideration. The exact curves are obtained via Monte-Carlo simulation. A lower bound is obtained using (13) and the average interference on the desired UE, $\bar{I}$, calculated through Propositions \ref{proposition1}-\ref{proposition3}.  
During the simulations 
the system for the HD scenario is considered to be FDD. See Tab.~\ref{tab:parameters} for the parameters adopted to carry out the simulations. 

In the following, Section \ref{Results on HD and FD HCNs} provides results on HD and FD HCNs and studies the impact of FD technology on the QoS experienced by the end users. In Section \ref{Validating the lower bound} the accuracy and complexity of the proposed lower bound are investigated. Finally in Section \ref{FD vs. HD Tradeoff} the crossover point under different QoS requirements is investigated, and the minimum level of SI cancellation required for the FD system to outperform the HD one is determined.

	
	\begin{table}[t]
		\center
		\caption{System and Simulation Parameters}
		\begin{tabular}{lll}
			\hline
			\cmidrule(r){1-3}
			Description    & Parameter & Value \\
			\midrule
			Macro BS TX Power       & $P_{\text {M-BS}}$         &  46 dBm      \\
			Small Cell BS TX Power        & $P_{\text {P-BS}}$         &  35 dBm      \\
			User TX Power            & $P_{\text {UE}}$             &  23 dBm      \\
			Service Time             & $T_f$                        &  0.5 ms      \\
			Service Bandwidth        & $BW$                         &  180 kHz     \\
			Path loss exponent       & $\alpha$                     &  3           \\
			Noise Power              & $ \sigma^2 $                 &  -120 dBm    \\
			Pico--Pico BSs Minimum Distance  &-                     &  180 meters    \\
			Macro--Pico BSs Minimum Distance &-                     &  80 meters    \\                
			Coverage Radii of Pico cells               &- 
			& 90 meters \\
			Non--linear SI cancellation parameter &$\kappa$          &  1           \\
			Linear SI cancellation parameter     &$\eta$            &  Not fixed   \\
			QoS exponent             &$\theta$                      &  Not fixed   \\
			Density of small BSs & $\lambda$ & Not fixed \\
			\bottomrule
		\end{tabular}
		\label{tab:parameters}
	\end{table}

\begin{figure*}[!t]
	\vspace{-0.25cm}
	\centering
	\subfigure[]{\includegraphics[width=0.42 \textwidth]{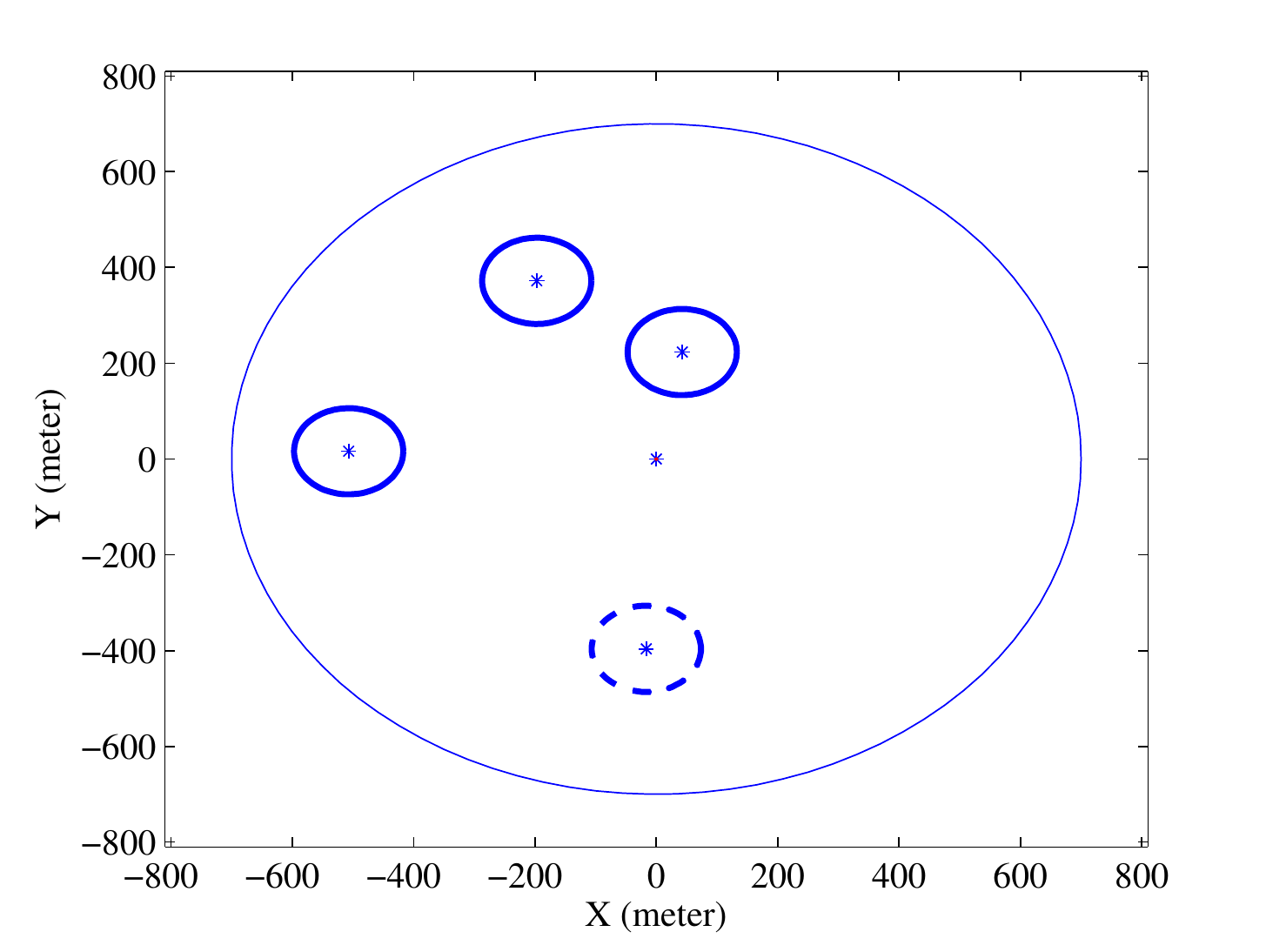}}
	\label{fig.10}
	\hfil
	\subfigure[]{\includegraphics[width=0.42 
		\textwidth]{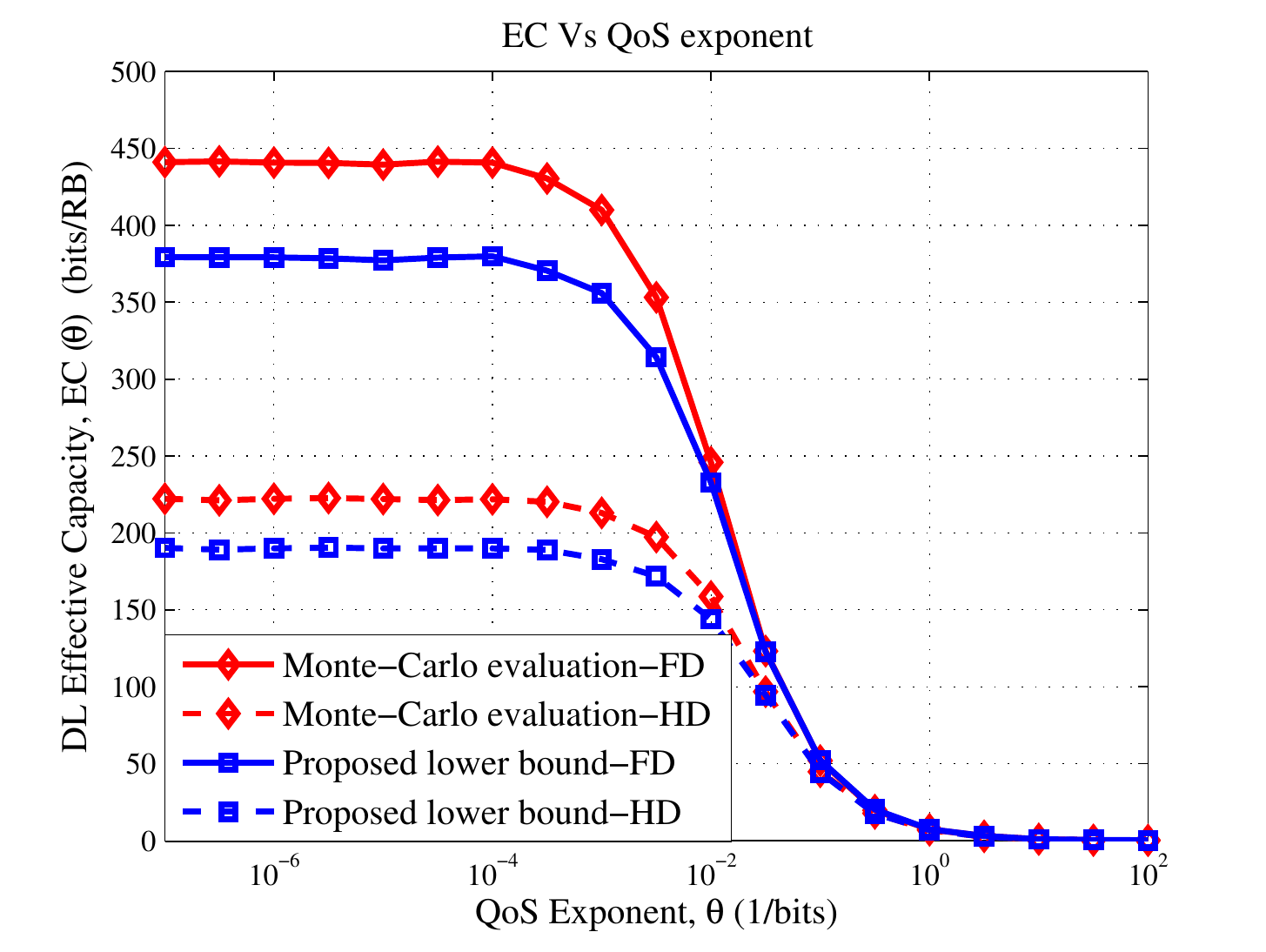}} 
	\caption{(a) A specific realization of sparse small cell deployment using a Mat\'ern HCPP with density $\lambda=5 \, 
		\text{small cells/km}{^2}$, and (b) DL effective capacity experienced by a typical UE (uniformly distributed in the dashed small cell) vs. QoS exponent, for HD and FD (exact and lower bounds) with $\eta = -150$ dB.}
	\label{Figure1}
\end{figure*}	
	
	\begin{figure*}[!t]
		\centering
		\subfigure[]{\includegraphics[width=0.42 \textwidth]{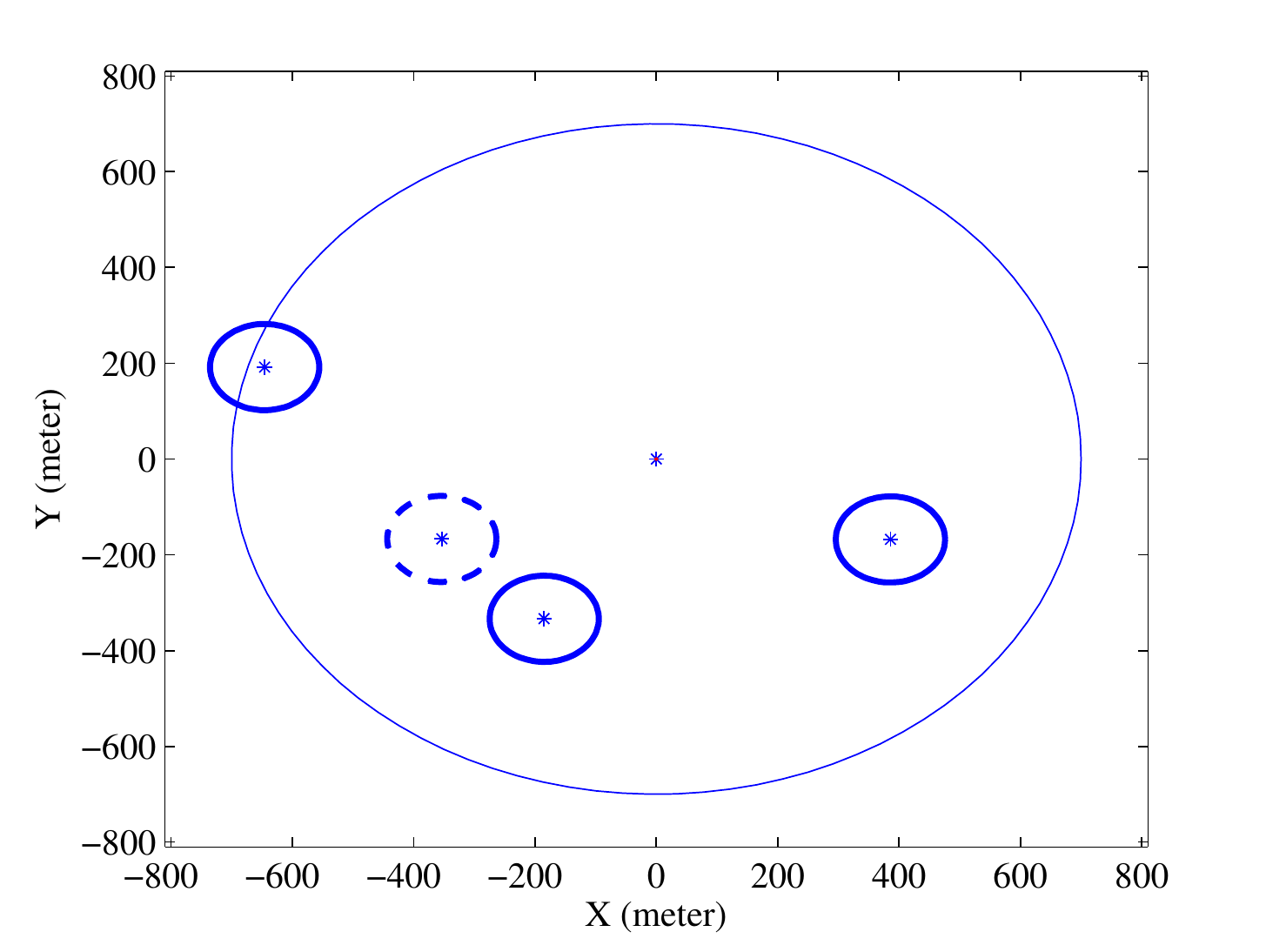}}
		\hfil
		\subfigure[]{\includegraphics[width=0.42 \textwidth]{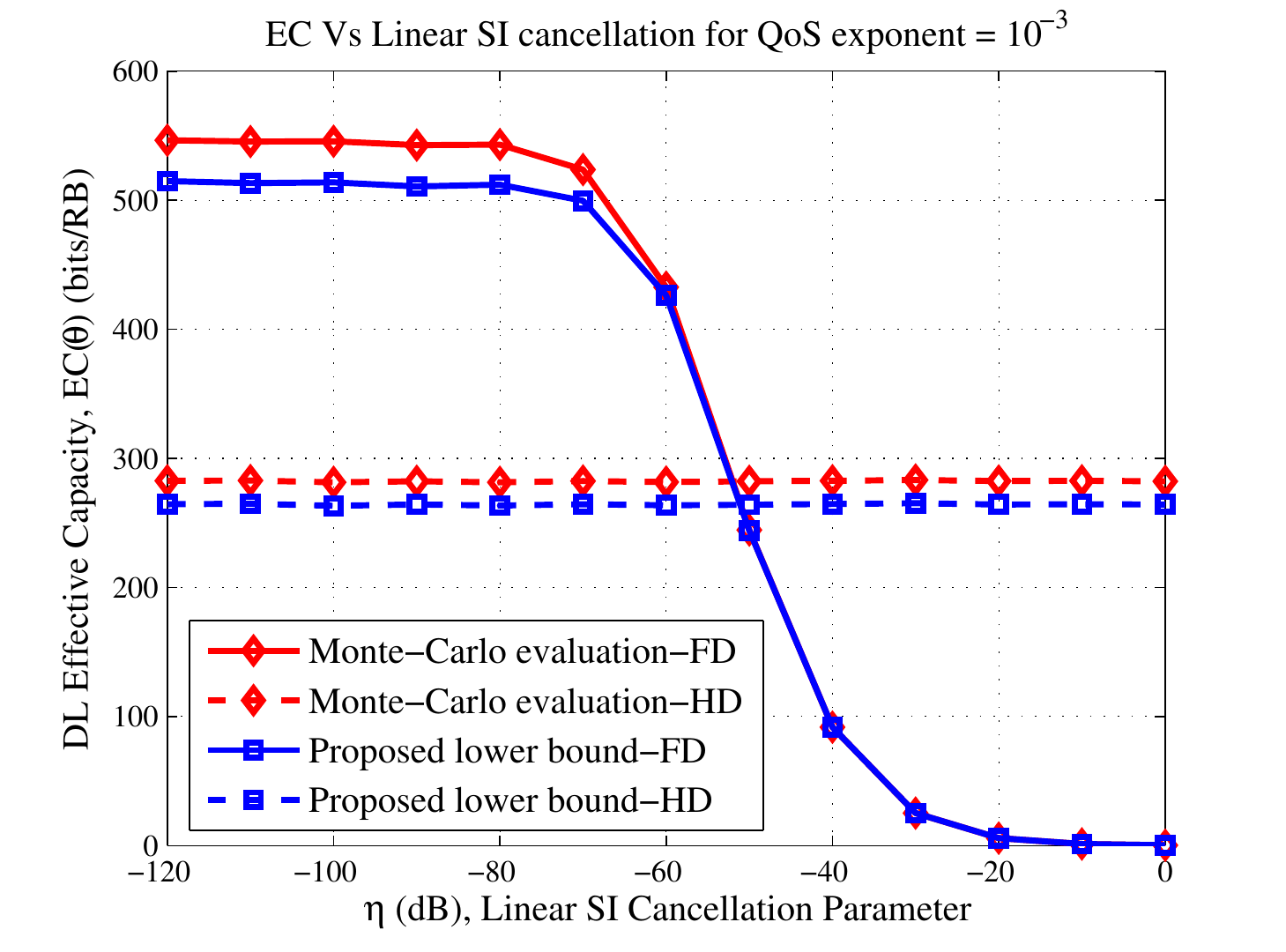}}
		\caption{(a) A specific realization of sparse small cell deployment 
		using a Mat\'ern HCPP with density $\lambda=5 \, 
			\text{small cells/km}$$^2$, and (b) corresponding DL effective 
			capacity experienced by a typical UE (uniformly distributed in the 
			dashed small cell) vs. linear SI 
			suppression ratio, for HD and FD (exact and lower bounds). QoS 
			exponent $\theta=10^{-3}$ (1/bit).}
		\label{Figure2}
	\end{figure*}
	
	\begin{figure*}[!t]
		\vspace{-0.25cm}
		\centering
		\subfigure[]{\includegraphics[width=0.42 \textwidth]{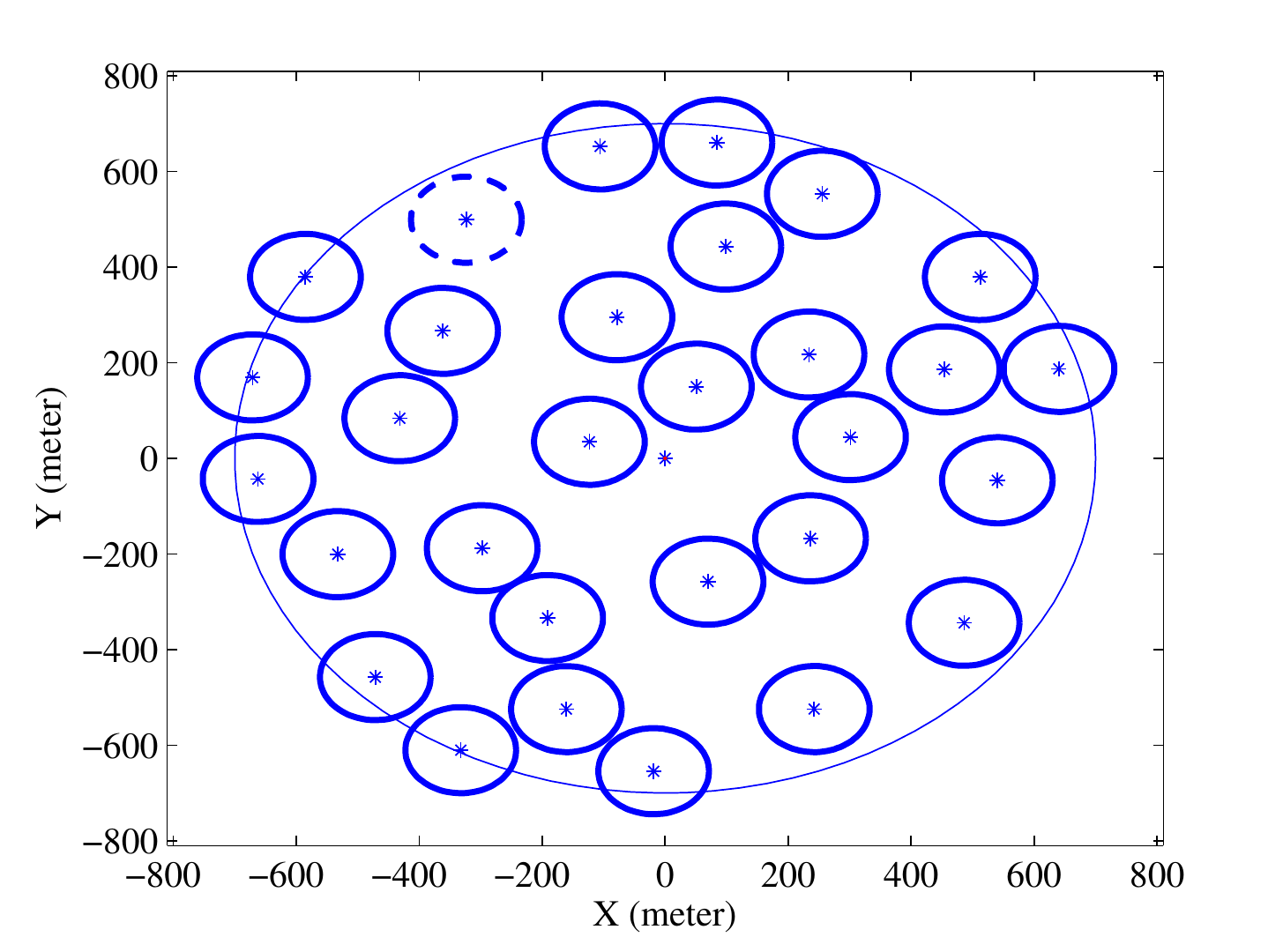}}
		\label{fig.8}
		\hfil
		\subfigure[]{\includegraphics[width=0.42
			\textwidth]{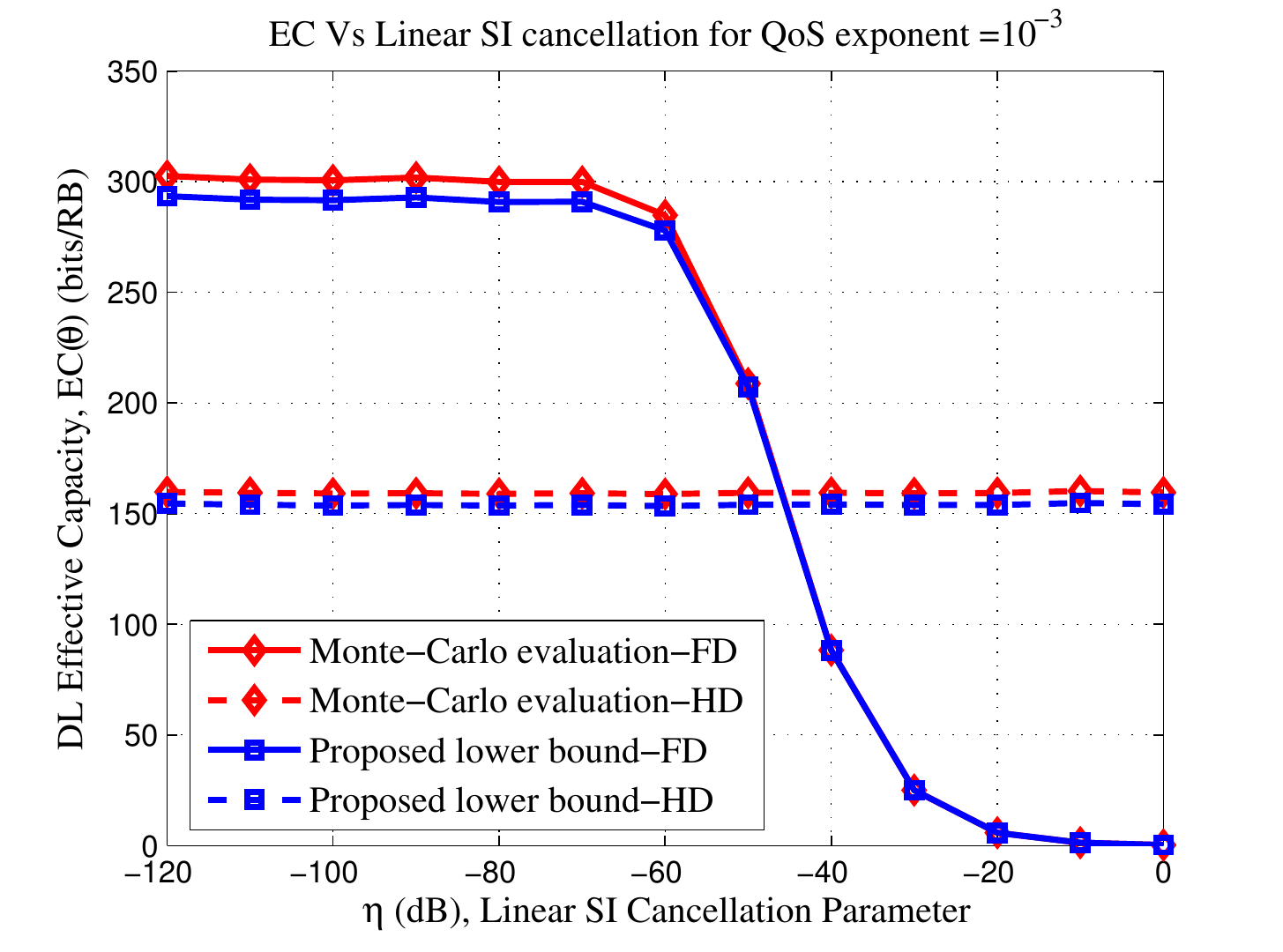}} 
		\caption{(a) A specific realization of dense small cell deployment using a Mat\'ern HCPP with density $\lambda=50 \, 
			\text{small cells/km}$$^2$, and (b) DL effective 
			capacity experienced by a typical UE (uniformly distributed in the dashed small cell) vs. linear SI 
			suppression ratio, for HD and FD (exact and lower bounds). QoS 
			exponent $\theta=10^{-3}$ (1/bit).}
		\label{Figure3}
	\end{figure*}



\subsection{Performance evaluation of HD and FD HCNs} 
\label{Results on HD and FD HCNs}

 Fig. \ref{Figure1} reports the DL EC of a perfect FD 
 system compared with that of a HD network for different QoS exponents. As it can be seen, a 1.95X gain from a perfect FD compared to HD 
 in terms of~EC is obtained. This is due to the fact that in completely FD systems there would be more interference in the network compared to a completely HD 
 one since UL and DL transmissions may occur in the same RBs. It is worth noting again that when $\theta \rightarrow 0$ there is no sensitivity at all to delays in the system, and in this case the EC tends to the Shannon capacity.

 Figs.~\ref{Figure2} and ~\ref{Figure3} show DL EC versus linear SI cancellation parameter, $\eta$, under sparse and dense HCN scenarios, respectively. As it can be seen, depending on the level of SI cancellation, the system performance under FD scenario exhibits various performance results. For instance, for low SI cancellation parameter, i.e., $\eta \rightarrow 0$ dB, the system can not provide any service to the UE. However, as the system goes towards perfect FD, i.e., $\eta \rightarrow -\infty$ dB, the system provides approximately 2X service to the UE compared to the HD counterpart. In addition, there is a minimum linear SI cancellation required so that FD operation mode outperform HD one. We call this crossover point and as it can bee seen, it is approximately $-50$~dB in this case. Detailed discussions on this crossover point and its dependence on the system parameters are presented in Sec.~\ref{FD vs. HD Tradeoff}.
 
 Moreover, in Fig. \ref{Figure2}.b the EC for the UE is 
 approximately 1.67X that observed in Fig. 
 \ref{Figure3}.b for the same level of SI cancellation. This is mainly 
 due to the lower level of interference in the network. However, the area spectral efficiency is much higher in the HCN of 
 Fig.~\ref{Figure3}. In fact, as the network becomes denser, a larger number of 
 UEs can be served simultaneously in a single RB. Thus, despite the 
 decreased service rate to each UE in a dense HCN, a very high area spectral efficiency can be 
 expected.

\subsection{Assessment of the Lower Bound: Accuracy and Complexity}
\label{Validating the lower bound}
As Figs.~\ref{Figure2} and~\ref{Figure3} demonstrate, the proposed lower bound is tight, and interestingly as the system becomes denser, the lower bound becomes tighter. Furthermore, according to Figs.~\ref{Figure2} and \ref{Figure3}, increasing the SI parameter, $\eta$, the gap between the lower bound and the exact value of EC decreases. This is due to the fact that as the SI parameter increases, the RSI dominates other sources of interference on the desired UE and since its power is constant as $\eta P^{\kappa}$, the interference tends to a constant and Jensen's inequality (invoked in \eqref{eq13}) becomes an equality. 

Fig.~\ref{Figure1}.b shows the impact of QoS exponent on the accuracy of the lower bound. As it is shown, under loose QoS requirements, i.e., $\theta \rightarrow 0$, the lower bound becomes looser, however, still it is tight enough to give insight on system performance. On the other hand, for stringent QoS requirements the lower bound becomes very tight. This figure demonstrates that in a very broad range of QoS exponent the proposed approach gives accurate results on the EC that a cellular network provides in DL for a UE.

\begin{figure}[!t]
	\vspace{-0.25cm}
	\centering
	\subfigure[]{\includegraphics[width=0.5 \textwidth]{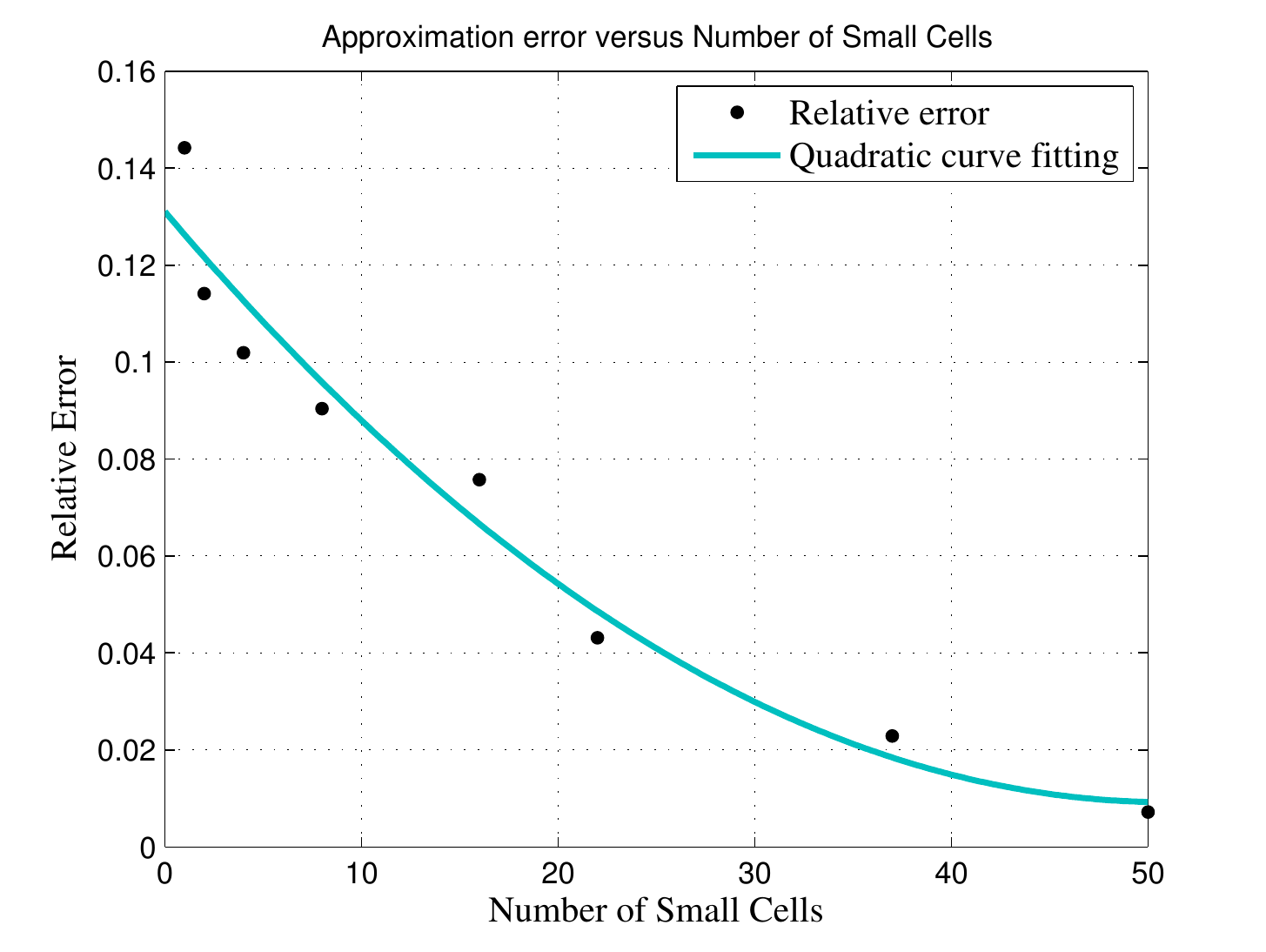}}
	\label{fig.11}
	\hfil
	\subfigure[]{\includegraphics[width=0.5 
		\textwidth]{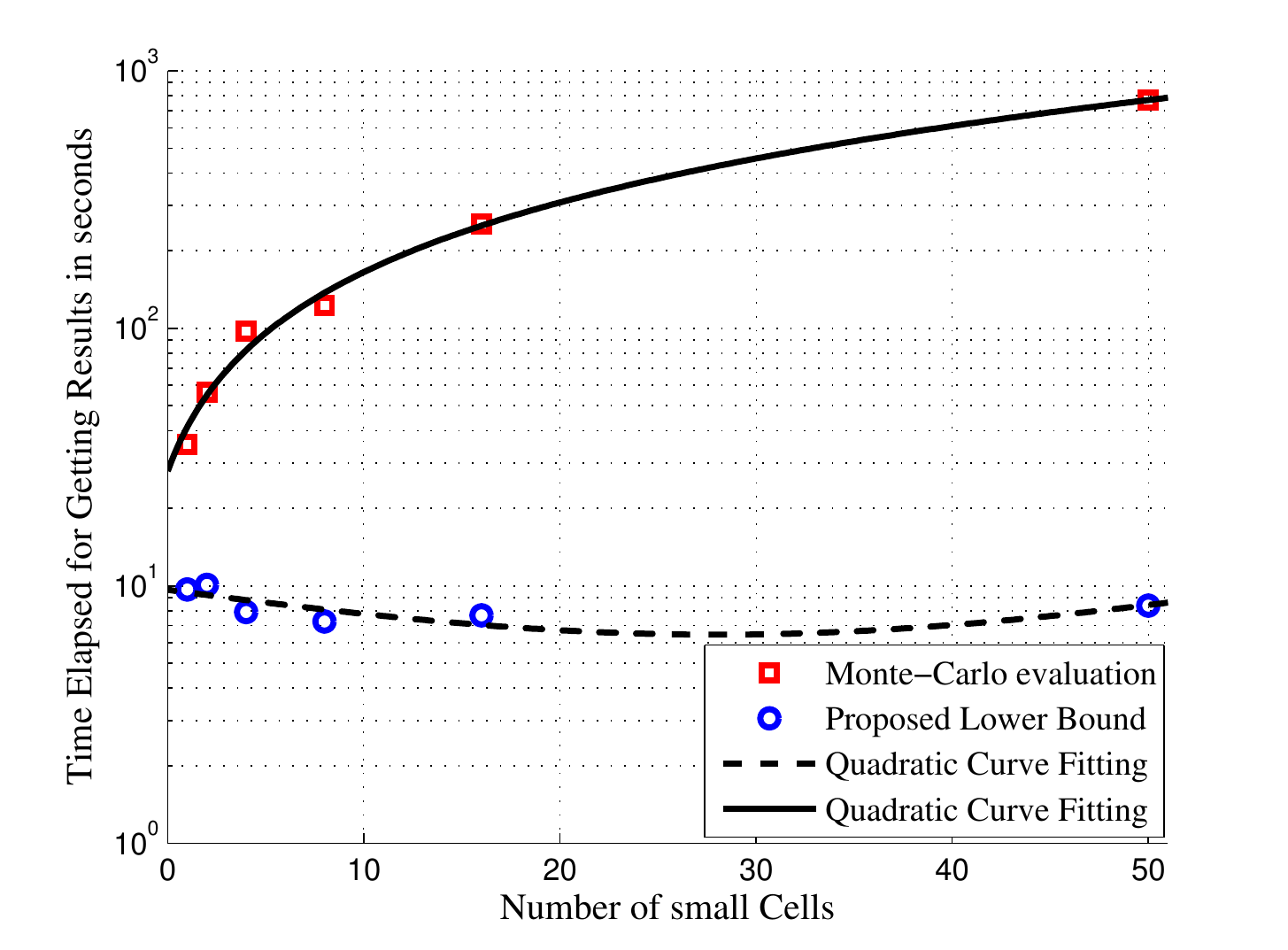}} 
	\caption[Accuracy and complexity analysis of the proposed scheme versus the number of small cells in the Macro cell averaged over 10 independent random network realizations. The parameters for these simulations are $\theta=10^{-5}$, and $\eta=-110dB$, $\kappa=1$.  Error of the lower bound (difference between exact and lower bound per RB) (a), the time elapsed to get the results (b).]{Accuracy and complexity analysis of the proposed scheme versus the number of small cells in a macro cell averaged over 100 independent network realizations. The parameters for these simulations are $\theta=10^{-5}$, and $\eta=-110$ dB, $\kappa=1$. (a)  relative error of the lower bound, and (b) time elapsed to get the results \footnotemark{ }.}
	\label{Accuracy_Complexity}
\end{figure}
\footnotetext{All the simulations were carried out with an Intel Core i5-2.53GHz processor and 4G RAM on a Dell Inspiron 5010.}

\begin{figure}[!h]
	\vspace{-0.25cm}
	\centering
	\includegraphics[width=0.5 \textwidth]{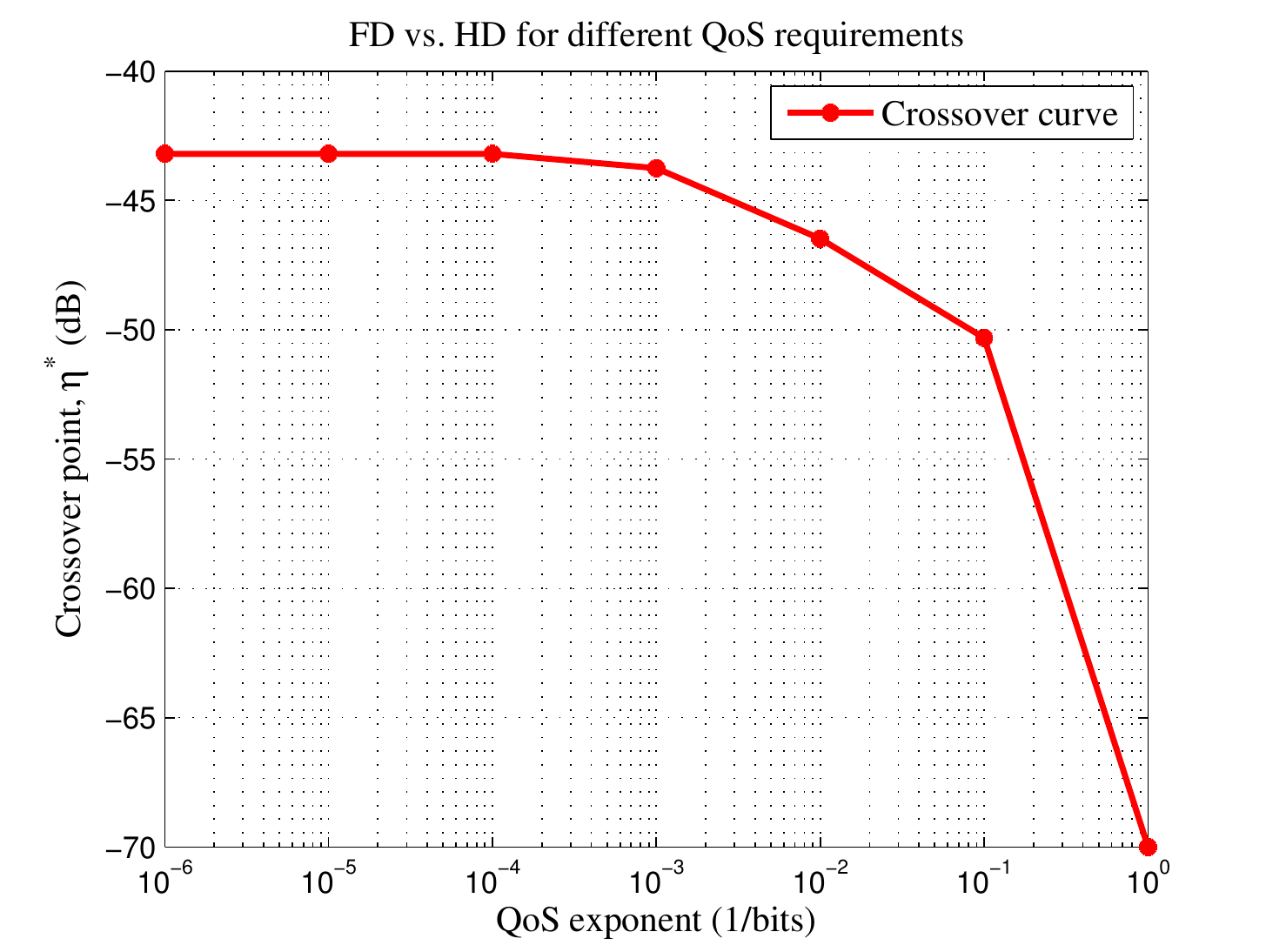}
	\caption{Crossover point, $\eta ^{\ast}$, versus QoS exponent, $\theta$, in a random realization of the HCN with $\lambda = 5$ small cells/km$^2$.}
	\label{Fig.10}
\end{figure}

Fig.~\ref{Accuracy_Complexity}.a depicts the relative error of the proposed lower bound versus the number of small cells. We considered the number of small cells from the set $\left\{1, 2, 4, 8, 16, 22, 37, 50\right\}$,  and averaged the  relative error of our scheme in 100 different realizations of the HCN deployment for each case. As shown in this figure, by increasing the number of interferes in the network, the lower bound becomes tighter.

As discussed earlier in Section \ref{sec:theoretical}, the complexity of the proposed approach for analyzing the system performance is almost independent of the size of the network and significantly smaller than that of the exact analysis. To be more concrete, an exact analysis has a complexity in the order of $\mathcal{O} (N^{M})$, where $N$ is the number of function evaluations in each dimension of the integration, and $M$ is the number of small cells in the network. In comparison, the computational complexity of the proposed lower bound is $\mathcal{O}(N+M)$, where the first term of the inner argument corresponds to integration in the desired signal's dimension and the second term corresponds to calculating the average interference. To asses this claim for the same realizations of the network considered in Fig.~\ref{Accuracy_Complexity}.a, the average time elapsed to get the results is provided in Fig. \ref{Accuracy_Complexity}.b as a metric of complexity. As the number of interferes in the network grows, the complexity of the exact analysis increases dramatically, while the proposed lower bound has complexity almost independent of the network size.

Given the complexity and accuracy issues discussed above, it can readily be claimed that the approach proposed in this paper is scalable. That is, the denser the network, the tighter the approximation, without any significant additional complexity.

\subsection{Crossover Point}
\label{FD vs. HD Tradeoff} 
The SI cancellation techniques play an essential role in enabling a wireless device to operate in FD mode; see e.g., \cite{InbandSabharwal, bharadia2013full, bharadia2014full,duarte2010full,ChoiAchiSingle,CS}. Despite the abundance of many different proposed RF, active, and passive interference cancellation techniques with a variety of performance results, establishing any FD wireless link requires carefully designed SI cancellation techniques specific to the required application. In this regards, a system level study helps gaining insight on the required level of SI cancellation and thus on finding the appropriate SI cancellation techniques. To this aim, in the following we investigate the crossover point, $\eta^{\ast}$, i.e., the minimum SI cancellation ($\eta$) required for the FD system to start outperforming the HD counterpart. This would require solving the equation:
\begin{align}
	\label{eq.29}
	\nonumber
	& EC_{HD}(\theta)=- \frac{1}{\theta }\log \left[ {{\E_{s,I}}{{\left( {1 + \textrm{SINR}{_\textrm{HD}}} \right)}^{ - \theta {T_f}\frac{{BW}}{2}{{\log }_2}e}}} \right] =  \\
	&- \frac{1}{\theta }\log \left[ {{\E_{s,I}}{{\left( {1 + \textrm{SINR}{_\textrm{FD}}} \right)}^{ - \theta {T_f}{{BW}}{{\log }_2}e}}} \right]= EC_{FD}(\theta)
\end{align}
where ${\textrm{SINR}{_{\textrm{HD}}}}$ and ${\textrm{SINR}{_{\textrm{FD}}}}$ are provided in \eqref{Dual-eq7} and \eqref{eq7}, respectively. Eq. \eqref{eq.29} is a complicated non-linear function of the parameter $\eta$, the topology of the network, the QoS constraint, $\theta$, and the transmit powers in the network, which makes finding any closed form expression for the trade off point not possible. However, an estimate of such a point can be found numerically.  

As an example{\color{blue},} we plotted this trade off point in Fig. \ref{Fig.10} for a random 
realization of the HCN. This figure indicates the relation between the crossover 
point, $\eta^{\ast}$, and the QoS exponent, $\theta$. 
{Not intuitively, this 
plot indicates that the more stringent the QoS constraint, the higher the SI 
cancellation level required in order to make FD preferable to HD. In other 
words, if the QoS constraint is higher, the impact is profound on a FD system 
compared to its HD counterpart.}

{{Finally, Figs. \ref{Fig.7} and \ref{EC_BSPOWER_100dB} provide the results for the  DL EC of a typical UE as a function of the transmit power of the small cell BSs. As these figures show, the transmit power for which the FD starts outperforming the HD counterpart is a function of the topology of the network and, more importantly, of the SI cancellation performance, $\eta, \; \kappa$. For instance, in Fig. \ref{Fig.7} this occurs at~35 dBm, while for higher SI cancellation parameters, e.g., $\eta = -100$~dB, FD would outperform HD in a wider transmit power ranges as shown in Fig. \ref{EC_BSPOWER_100dB}.}}

\begin{figure}[!t]
	\vspace{-0.25cm}
	\centering
	\includegraphics[width=0.5 \textwidth]{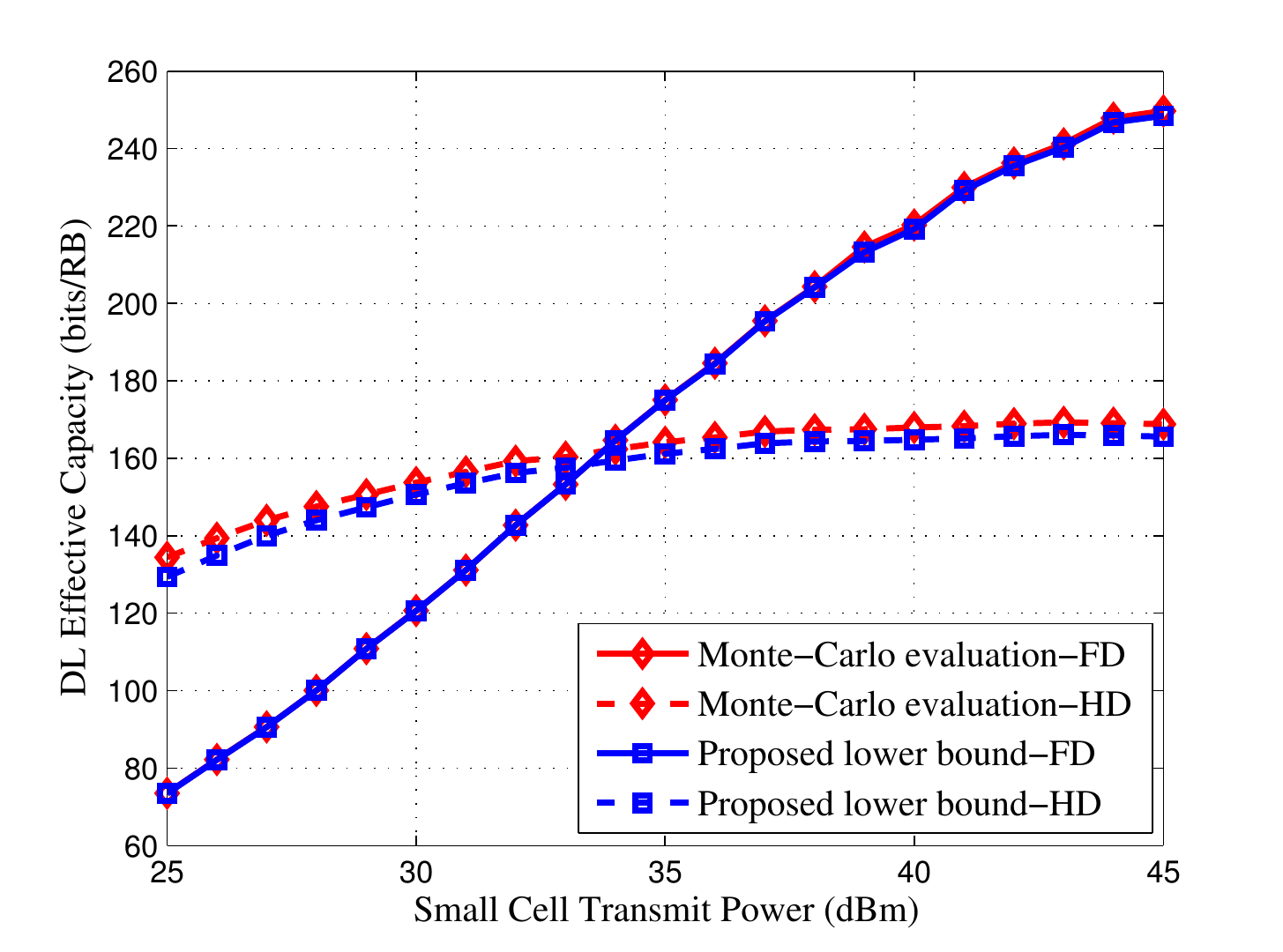}
	\caption{DL EC versus small cell transmit power for a random realization of the HCN, with $\eta= -50$ dB and $\theta = 10^{-2}$, $\lambda=30 \, 
		\text{small cells/km}{^2}$.}
	\label{Fig.7}
\end{figure}

\begin{figure}[!t]
	\vspace{-0.25cm}
	\centering
	\includegraphics[width=0.5 \textwidth]{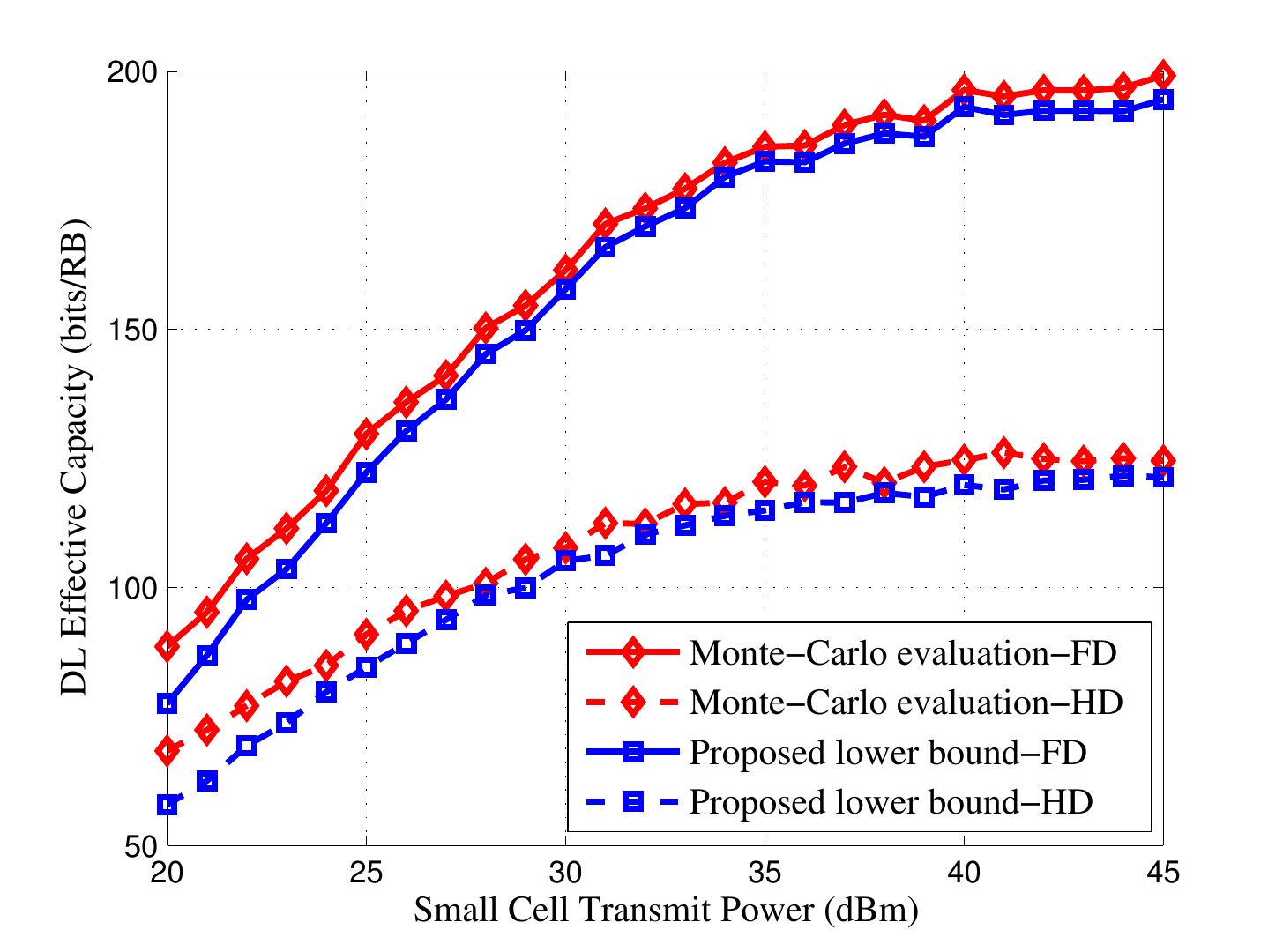}
	\caption{DL EC versus small cell transmit power for a random realization of the HC, with $\eta= -100$ dB and $\theta = 10^{-2}$, $\lambda=30 \, 
		\text{small cells/km}{^2}$.}
	\label{EC_BSPOWER_100dB}
\end{figure}

\section{Discussions}
\label{Discussions}

\subsection{ UE Distribution, Coverage Areas, and Channel Models}
	\label{UE distribution, Coverage areas}
	
	The uniform distribution of the UEs around their corresponding BSs, which is a result of the Mat\'ern HCPP model for the system, can be relaxed. This is due to the fact that the result given in (\ref{eq13}) is true for \textit{any} 
	distribution of the UEs and BSs in the network. 
	In fact, we can expect to find similar results for other point processes 
	that can be used to model the nodes. For instance, we can expect the 
	same procedure to be applied also for Thomas point processes, where child points 
	(UEs) are distributed with an isotropic Gaussian distribution with a given 
	variance around the cluster heads (BSs) \cite{StochGeoEk}. In fact, for any generic 
	distribution we only need to build a computationally efficient procedure to 
	find the average interference on the desired network entity (UE or BS), 
	as we have done in Section \ref{subsec:average_interference}.
	
	
	Finally, we assumed a pathloss dominated AWGN Rayleigh fading channel model, but, since 
	the lower bound presented in (\ref{eq13}) is always valid for any distribution 
	of the interferer's channel, it is straightforward to extend this work to other channel models, e.g., including log-normal shadowing.

\subsection{Overlapping Small Cells}
\label{Overlapping Small Cells}

{The analysis carried out in Section~\ref{sec:theoretical} relies on the 
assumption that two small cells deployed in the same macro cell do not overlap. 
However, in very dense network deployments, small cells are not perfectly 
separated and may overlap in some cases, thus impairing the effectiveness of 
the analysis proposed in this paper.}

{In detail, when two small cells overlap, as depicted in 
Fig.~\ref{Fig8}, there can be two users deployed in the overlapping regions 
which are served in the same RB, since we did not make any assumption on the 
resource scheduling strategy. In this case, the interference caused by the 
neighboring UE to the free UE (whose average value is expressed by 
Proposition~(\ref{proposition2})) is very high, since the distance between the 
two UEs is small. This is particularly relevant if 
both UEs are FD and, hence, both are transmitting and receiving at the same 
time. As a consequence, the term ${E_I}\left( I \right) = \bar I$ may become very 
high in this scenario and, consequently, the lower bound proposed in 
\eqref{eq13} is much less tight.}

{Fortunately, in a real network scenario, the possibility that two UEs 
very close to each other are served in the same RB is practically very small, and can be neglected. 
Indeed, exploiting cell coordination and channel measurements, two close UEs 
deployed in different but overlapping cells will be scheduled in different RBs. 
Consequently, the theoretical approach proposed in this paper is of 
practical interest.}

\begin{figure}[!t]
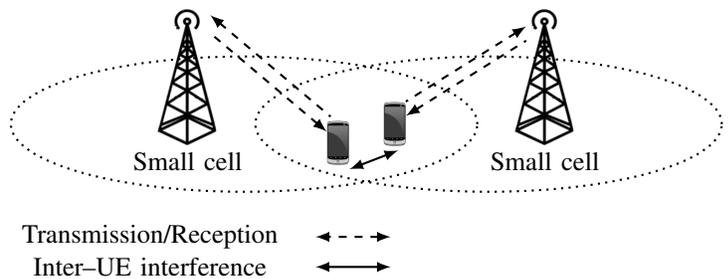

	\centering
	\begin{tikzpicture}

	\node at (-3.5,+.65) {\includegraphics[scale=0.19]{macro_cell.png}};
	\node at (-3.5,-.5) {Small cell};
	
	\node at (+1.25,+.65) {\includegraphics[scale=0.19]{macro_cell.png}};
	\node at (1.25,-.5) {Small cell};
	
	\draw[latex-,draw=black,thick,dotted] (+.5,0) ellipse (3.1cm and .9cm);
	\draw[latex-,draw=black,thick,dotted] (-2.75,0) ellipse (3.1cm and .9cm);
	\node at (-1.5,-.25) {\includegraphics[scale=0.08]{ue.png}};
	
	\node at (-.75,0) {\includegraphics[scale=0.08]{ue.png}};
	
	\draw[latex-,draw=black,thick,dashed] (-3.2,1.45) -- (-1.6,.1);	
	
	\draw[latex-,draw=black,thick,dashed] (-1.65,-.1) -- (-3.2,1.21);	
	
	\draw[latex-,draw=black,thick,dashed] (1,1.3) -- (-.6,.3);	
	
	\draw[latex-,draw=black,thick,dashed] (-.6,.11) -- (1,1.11);
	
	\draw[latex-latex,draw=black,thick] (-.7,-.35) -- (-1.35,-.6);
	
	

	\node at (-4,-1.5) {Transmission/Reception};
	\node at (-4,-1.9) {Inter--UE interference};
	
	\draw[latex-latex,draw=black,thick,dashed] (-.8,-1.5) -- (-1.8,-1.5);
	
	\draw[latex-latex,draw=black,thick] (-.8,-1.9) -- (-1.8,-1.9);
	
	\end{tikzpicture}
	\caption{{Inter--UE interference in overlapping small cell scenario 
			(full duplex).}}
	\label{Fig8}
\end{figure}

	\section{Conclusions}
	\label{sec:conclusions}
	
	In this paper we introduced a lower bound for the evaluation of EC in a  wireless cellular scenario based on which we built a scalable mathematical framework to analyze the statistical QoS
	performance of ultra dense next generation HCNs. Our proposed scalable 
	scheme helped us analyzing HD and imperfect FD HCNs from
	the EC perspective with high accuracy at only a fraction of
	the complexity needed for an exact analysis. Moreover, we investigated the trade off between FD and HD systems by finding the minimum amount of SI cancellation needed for FD HCN to perform similar to HD in terms of EC.

\appendices

\section{}
\label{AppB}
The average path loss between the free UE and the desired UE based on Fig. \ref{fig2.5} can be found by computing

\begin{eqnarray}
\E\left( {{{\left\| {x'} \right\|}^{ - \alpha }}} \right) = \E{\left( {{r_1'}^2 + {{c'}^2} - 2c'{{r_1'}}\cos \left( {{{\theta_1'}} - \psi '} \right)} \right)^{ - \frac{\alpha }{2}}}
\label{AppBeq3}
\end{eqnarray}
where
\begin{eqnarray}
{{c'}^2} = {{d'}^2} + {r'_2}^2 + 2{{r_2'}}d'\cos {{\theta'}_2} \label{AppBeq2}
\end{eqnarray}
To approximately calculate (\ref{AppBeq3}) this time we assume that $\left( {{{r_2'}},{{\theta '}_2}} \right)$, i.e., the position of the desired UE, is fixed, and then take the expectation with respect to the position of the free UE and finally with respect to the desired UE.
\begin{align}
\nonumber
&\,\,{\E_{\left( {{{r_1'}},{{\theta_1'}}} \right)}}\left[ {\left. {{{\left( {{r_1'}^2 + {{c'}^2} - 2c'{{r_1'}}\cos \left( {{{\theta_1'}} - \psi '} \right)} \right)}^{ - \frac{\alpha }{2}}}} \right|\left( {{{r_2'}},{{\theta_2'}}} \right)} \right] \\ \nonumber
&= \int\limits_0^{R_1'} {\int\limits_{{\theta ^ * }}^{2\pi  - {\theta ^ * }} {{{\left( {{r'_1}^2 + {{c'}^2} - 2c'{{r_1'}}\cos \left( {{{\theta_1'}} - \psi '} \right)} \right)}^{ - \frac{\alpha }{2}}}} } \frac{{{{r_1'}d{{r_1'}}d{{\theta_1'}}}}}{{(\pi-\theta ^ *) {{R_1 '}^2}}}\\ \nonumber
&\mathop  = \limits^{(a)} \underbrace {\int\limits_0^{c'} {\int\limits_{{\theta ^ * }}^{2\pi  - {\theta ^ * }} {{{\left( {{r_1'}^2 + {{c'}^2} - 2c'{{r_1'}}\cos \left( {{{\theta_1'}} - \psi '} \right)} \right)}^{ - \frac{\alpha }{2}}}} } \frac{{{r_1'}d{{r_1'}}d{{\theta_1'}}}}{{(\pi-\theta ^ *) {{R_1 '}^2}}}}_{{I_1}} \\ \nonumber
&+ \underbrace {\int\limits_{c'}^{R_1'} {\int\limits_{{\theta ^ * }}^{2\pi  - {\theta ^ * }} {{{\left( {{r'_1}^2 + {{c'}^2} - 2c'{{r_1}'}\cos \left( {{{\theta_1 '}} - \psi '} \right)} \right)}^{ - \frac{\alpha }{2}}}} } \frac{{{{r_1}'}d{{r_1'}}d{\theta_1 '}}}{{(\pi-\theta ^*) {R_1^2}}}}_{{I_2}} \\
\label{APPB_4}
\end{align}
where (a) holds, since $\left( {{{r_2'}},{{\theta_2'}}} \right)$ is fixed and consequently $ \left( c',\psi '\right) $ becomes constant, thus the integration can be separated into two pieces. 

By factorizing $c'$ in the first integral and ${r_1'}$ in the second, we can find an approximation for (\ref{APPB_4}) due to the fact that ${{r_1'}} < c'$ holds always in the former while $c'< {{r_1'}}$ in the latter.
\begin{align}
\nonumber
& I_{1}=\frac{{{{c'}^{ - \alpha }}}}{{{{R_1'}^2}{(\pi-\theta ^ *) }}} \times \\ \nonumber
&\,\,\,\,\,\,\,\int\limits_0^{c'} {\int\limits_{{\theta ^ * }}^{2\pi  - {\theta ^ * }} {{{\left( {1 + {{\left( {\frac{{{{r_1'}}}}{{c'}}} \right)}^2} - 2\frac{{{{r_1'}}}}{{c'}}\cos \left( {{{\theta_1'}} - \psi '} \right)} \right)}^{ - \frac{\alpha }{2}}}} } {{r_1'}}d{{r_1'}}d{{\theta _1'}} \\ \nonumber
&\approx \frac{{{{c'}^{ - (\alpha  - 2)}}}}{{{{R_1'}^2}}}\left\{ {1 - \left[ {\frac{\alpha}{4} + \frac{2\alpha}{{3\left( {\pi  - {\theta ^*}} \right)}}\sin {\theta ^*}\cos \psi '} \right]} + \frac{{\alpha \left( {\alpha  + 2} \right)}}{8} \right. \\ 
&\left. \times { \left[ {\frac{4}{3} + \frac{1}{{\left( {\pi  - {\theta ^*}} \right)}}\left( {\frac{8}{5}\sin {\theta ^*}\cos \psi ' - \frac{1}{2}\sin 2{\theta ^*}\cos 2\psi '} \right)} \right]} \right\}
\label{AppBeq5}
\end{align}
In the second integral, again by using the same  approach as in (\ref{AppBeq5}), we obtain
\begin{align}
\nonumber
& I_2= \frac{1}{{{R_1'}^2}{{(\pi-\theta ^ *) }}} \times \\ \nonumber
& \,\,\,\,\int\limits_{c'}^{R_1'} {\int\limits_{{\theta ^ * }}^{2\pi  - {\theta ^ * }} {{{r_1'}}^{\,\left( { - \alpha  + 1} \right)}{{\left( {1 + {{\left( {\frac{{c'}}{{{{r_1'}}}}} \right)}^2} - 2\frac{{c'}}{{{{r_1'}}}}\cos \left( {{{\theta_1 '}} - \psi '} \right)} \right)}^{ - \frac{\alpha }{2}}}} } \\ \nonumber
& \times d{{r_1'}}d{{\theta_1'}} \\ \nonumber
& \approx \frac{1}{{{{R_1'}^2}}}\left\{ { {\frac{2}{{\alpha  - 2}}} \left({{c'}^{ - \left( {\alpha  - 2} \right)}- {{{R_1'}^{ - \left( {\alpha  - 2} \right)}}}}\right) - \frac{\alpha }{2}\left[ {\frac{{2}}{\alpha}}{{c'}^2} \left( {{{c'}}^{ - \alpha }} \right. \right.} \right.\\ \nonumber
& \left. { - {{R'}_1}^{ - \alpha }} \right) \left. {+ \frac{{4c'}}{{(\alpha  - 1)}(\pi -\theta ^ *)}\left( {{{c'}^{ - \alpha  + 1}}-{{R_1'}^{ - \alpha  + 1}}} \right)\sin {\theta ^ * }\cos \psi '} \right]\\ \nonumber
& + \frac{{\alpha \left( {\alpha  + 2} \right)}}{8}\left[ {\frac{{2{{c'}^4}}}{{\alpha  + 2}}\left( {{{c'}^{ - \left( {\alpha  + 2} \right)}} - {{R_1'}^{ - \left( {\alpha  + 2} \right)}}} \right)} +\frac{{4{{c'}^2}}}{\alpha } \times \right.\\ \nonumber
&\,\, \left( {{{c'}^{ - \alpha }} - {{R_1'}^{ - \alpha }}} \right) - \frac{{2{{c'}^2}}}{{\alpha (\pi  - {\theta ^ * })}}\left( {{{c'}^{ - \alpha }} - {{R_1'}^{ - \alpha }}} \right)\sin 2{\theta ^ * }\cos 2\psi ' \\ 
&\left. {\left. {+ \frac{{8{{c'}^3}}}{{\left( {\alpha  + 1} \right)(\pi  - {\theta ^ * })}}\left( {{{c'}^{ - \left( {\alpha  + 1} \right)}} - {{R_1'}^{ - \left( {\alpha  + 1} \right)}}} \right)\sin {\theta ^*}\cos \psi '} \right]} \right\}
\label{AppBeq6}
\end{align}
We can simplify (\ref{AppBeq6}) by recalling that ${R_1'}$ is the radius of the coverage area of the macro BS, thus we have $ {{R_1'}^{ - t }} \ll {{c'}^{ - t }}, \forall t>0$. By taking this into account and using (\ref{AppBeq6}) and (\ref{AppBeq5}) we obtain the approximation for the average path loss from a free UE to the desired UE as
	\begin{align}
		\nonumber
		& {I_{{\text{free UE - UE}}}} \approx {P_{{\rm{UE}}}}. \\ \nonumber
		&\left[ {\left( {\frac{{\left( {\alpha  + 2} \right)\left( {\alpha  + 3} \right)}}{6} + \frac{2}{{\alpha  - 2}}} \right)} \right.\frac{1}{{{{R'_1}^2}}}\E\left( {{{c'}^{ - (\alpha  - 2)}}} \right) \\ \nonumber
		&- \frac{{\alpha \left( 3 \alpha^3 +11 \alpha^2-18 \alpha -56 \right)}}{{15(\pi  - {\theta ^*})\left( {{\alpha ^2} - 1} \right)}}\frac{{\sin {\theta ^*}}}{{{{R'_1}^2}}}\E\left( {{{c'}^{ - (\alpha  - 2)}}\cos \psi '} \right) \\ 
		&\left. { - \frac{{\left( {\alpha  + 2} \right)\left( {\alpha  + 4} \right)}}{{16(\pi  - {\theta ^*})}}\frac{{\sin 2{\theta ^*}}}{{{{R'_1}^2}}}\E\left( {{{c'}^{ - (\alpha  - 2)}}\cos 2\psi '} \right)} \right].
		\label{AppBeq7}
	\end{align}
where $\E\left( {{{c'}^{ - (\alpha  - 2)}}} \right)$ is already provided in (\ref{eq17}) and $\E\left( {{{c'}^{ - (\alpha  - 2)}}\cos \psi '} \right)$ and $\E\left( {{{c'}^{ - (\alpha  - 2)}}\cos 2\psi '} \right)$ can be calculated as follows 

\begin{align}
\nonumber
& \E\left( {{{c'}^{ - \left( {\alpha  - 2} \right)}}\cos \psi '} \right) = \\ \nonumber 
& d'  \E\left( {{{c'}^{ - \left( {\alpha  - 1} \right)}}} \right) + \E\left( {{{c'}^{ - \left( {\alpha  - 1} \right)}}{{r'}_2}\cos {{\theta '}_2}} \right) \approx \\ 
& d' \E\left( {{{c'}^{ - \left( {\alpha  - 1} \right)}}} \right) - \frac{{\left( {\alpha  - 1} \right){{d'}^{ - \alpha }}R'{{_2^2}}}}{4} + \frac{{\left( {{\alpha ^2} - 1} \right){{d'}^{ - \left( {\alpha  + 2} \right)}}R'{{_2^4} }}}{{12}}
\label{Ec(cospsi)}
\end{align}
and
\begin{align}
\nonumber
& \E\left( {{{c'}^{ - \left( {\alpha  - 2} \right)}}\cos 2\psi '} \right) = \\ \nonumber
& \E\left( {{{c'}^{ - \left( {\alpha  - 2} \right)}}\left( {1 - 2{{\left( {\frac{{{{r'}_2}\sin {{\theta '}_2}}}{{c'}}} \right)}^2}} \right)} \right) \approx \\ \nonumber
&  \E\left( {{{c'}^{ - \left( {\alpha  - 2} \right)}}} \right) - \left\{ {\frac{{{{d'}^{ - \alpha }}R'{{_2^2}}}}{2} + \alpha \left( {4\alpha  + 7} \right)\frac{{{{d'}^{ - \left( {\alpha  + 2} \right)}}R'{{_2^4}}}}{6}} \right.\\
& \left. {\,\,\,\,\, + \alpha \left( {\alpha  + 2} \right)\frac{{{{d'}^{ - \left( {\alpha  + 4} \right)}}R'{{_2^6} }}}{{32}}} \right\}
\label{Ec(cos2psi)}
\end{align}
In deriving (\ref{Ec(cospsi)}) and (\ref{Ec(cos2psi)}) we assumed $R'_2<d'$.




%
\bibliographystyle{IEEEtran}
\bibliography{References}

%

%

\end{document}